\documentclass[journal,comsoc]{IEEEtran}
\usepackage{latexsym}
\usepackage{graphicx}
\usepackage{amsfonts,amssymb,amsmath}
\usepackage{hyperref}

\usepackage[numbers,sort&compress]{natbib}
\usepackage[T1]{fontenc}
\usepackage{graphicx}
\usepackage[export]{adjustbox}
\usepackage[super]{nth}
\usepackage{diagbox} 
\usepackage{tabularx}
\usepackage{mathtools}
\usepackage{xcolor}

\usepackage{caption}
\usepackage{subcaption}

\usepackage[pdf]{pstricks}

\usepackage{amsmath}
\usepackage{setspace}
\usepackage[cmintegrals]{newtxmath}

\usepackage{textcomp, gensymb}

\usepackage{amsthm}
\makeatletter
\def\th@plain{%
  \thm@notefont{}
  \itshape 
}
\def\th@definition{%
  \thm@notefont{}
  \normalfont 
}

\makeatother
\newtheorem{theorem}{Theorem}

\newtheorem{remark}{Remark}
\newtheorem{lemma}{Lemma}

\newcommand*{\Perm}[2]{{}^{#1}\!P_{#2}}%

\begin{document}

\title{Intelligent Reflecting Surface Networks with Multi-Order-Reflection Effect: System Modelling and Critical Bounds}

\author{
Yihong Liu,~\IEEEmembership{Graduate Student Member,~IEEE},
Lei Zhang,~\IEEEmembership{Senior Member,~IEEE}, 
Feifei Gao,~\IEEEmembership{Fellow,~IEEE}\\
and Muhammad Ali Imran,~\IEEEmembership{Senior Member,~IEEE},
\thanks{ Yihong Liu, Lei Zhang and Muhammad Ali Imran are with the James Watt School of Engineering, University of
Glasgow, Glasgow, G12 8QQ, UK. (e-mail: y.liu.6@research.gla.ac.uk;
lei.zhang@glasgow.ac.uk; muhammad.imran@glasgow.ac.uk).

Feifei Gao is with the Department of
Automation, Tsinghua University, Beijing, China (e-mail: feifeigao@ieee.org).

This paper has been presented in part at the IEEE International Conference on Electronic Information and Communication Technology (ICEICT) 2021 \cite{myconfer2}.}

}

\maketitle
\begin{abstract}
In this paper, we model, analyze and optimize the multi-user and multi-order-reflection (MUMOR) intelligent reflecting surface (IRS) networks. We first derive a complete MUMOR IRS network model that is applicable for arbitrary times of reflections, arbitrary size and number of IRSs/reflectors. The optimal condition for achieving sum-rate upper bound with one IRS in a closed-form function and the analytical condition to achieve interference-free transmission, are derived respectively. Leveraging this optimal condition, we obtain the MUMOR sum-rate upper bound of IRS network with different network topology, where the linear graph (LG), complete graph (CG) and null graph (NG) topologies are considered. Simulation results verify our theories and derivations and demonstrate that the sum-rate upper bounds of different network topologies is under a $K$-fold improvement given $K$-piece IRS.




\end{abstract}


\begin{IEEEkeywords}
Intelligent reflecting surfaces networks, beamforming, MIMO, multi-order-reflection, sum-rate, graph theory.
\end{IEEEkeywords}

\IEEEpeerreviewmaketitle


%
\IEEEpeerreviewmaketitle

\section{Introduction}
The \nth{5} generation (5G) communication is supported by various radio and network techniques such as millimeter wave (mmWave), ultra-dense network, and massive multiple-input multiple-output (MIMO) to achieve unrivaled data rate, ultra-reliability, ultra-low latency communications, and satisfy the ever-increasing demands from various applications \cite{2017shafi5g}. Nevertheless, researchers have begun to seek the pathway towards the future \nth{6} generation (6G) communication, for obtaining even higher spectral efficiency (SE) and energy efficiency (EE).

The intelligent reflecting surface (IRS) \cite{gong2020survey}, also named as reconfigurable intelligent reflecting surface (RIS) \cite{yuan2021reconfigurable,tang2021pathloss} or metasurface \cite{yang2016,Zhao2020}, has been proposed as a potential 6G technique. The initial idea of IRS is originated from creating a smart and programmable wireless channel with a class of artificial surfaces. It can be produced by integrating artificially designed electronic elements, e.g., PIN diodes or varactors, on the facet of surfaces, e.g., printed circuit board (PCB), plus corresponding processors and controllers \cite{dai2019Relization}. The processor can compute the controlling parameters for reconfiguring each element based on different design criteria. The controller, e.g., field programmable gate array (FPGA), can correspondingly reconfigure the statement of each element \cite{zhangnature}. Then, the phase and amplitude of the reflected electromagnetic (EM) wave impinging on IRS can be manipulated correspondingly with designed manners. In this way, IRS is able to realize passive beamforming between transmitters (\texttt{Txs}) and receivers (\texttt{\texttt{Rxs}}) by reflecting the signal towards the desired \texttt{Rxs}, which essentially collect extra transmitted power from \texttt{Txs} to \texttt{Rxs}. Therefore, EE can be improved by using IRS to increase the signal-to-noise ratio (SNR) \cite{huang2018ee,zihan2020}. Meanwhile, IRS can suppress the inter-user interference by adding the interference power destructively at \texttt{Rxs} \cite{Liu2019}. From this view, more transceivers can share the same frequency bandwidth to achieve better SE \cite{Jiao2021}. In addition, compared with base stations (BSs) or active relaying (AF), IRS has a significantly lower cost because it does not involve any energy starving components like RF chains \cite{Wu2020}. 
A single IRS assisted communication systems have been considered in many works from different aspects, including EE maximization and weighted sum-rate maximization \cite{Wu2019,xie2022gao,Zhou2020}. An IRS network, which is defined as deploying multi-piece IRS in the transmission environment, has been studied to further enhance the EE and SE. In \cite{Kishk2021}, the statistical path-loss model of a large-scale IRS network is derived. The throughput of a single user (SU) has been maximized by IRS network leveraging the supervised learning approach \cite{Alexandropoulos2020}. Multi-user (MU) transmission via IRS network is investigated, considering minimizing the power consumption of transmit beamforming with constraints of the power supply, signal to average interference plus noise ratio (SINR) of each \texttt{Rx}, and constant modulus \cite{Sun2020}. The authors of \cite{mei2021mush} derived the lower bound of the MU average SINR by considering rayleigh fading channel in the IRS network \cite{mei2021mush}. The wideband transmission of MU has further been designed to maximize the sum-rate with limited power and constant modulus constraints in the IRS network \cite{zijian2021}. To realize decentralized IRS network, the authors of \cite{shaocheng2020} proposed distributed scheme of IRS network to maximize the MU weighted sum-rate. Additionally, the IRS network has been proposed to realize robust, secure MU communication by jointly designing the transmit beamforming, artificial noise and IRS network \cite{Yu2020}. Considering the multi-order-reflection (MOR) \cite{conway1994measurement,Maltsev2009,tam1995raytrace}, the authors of \cite{Mei2021} analyzed the single user multi-order-reflection (SUMOR) transmission in one path of the IRS network and then provided the beam routing solution. Further, the authors gave a tutorial for optimizing the wireless channel of one reflection to multi-user and multi-order-reflection (MUMOR) transmission \cite{mei2021intelligent}.

However, an intrinsic nature of EM wave transmission in IRS has been overlooked in the literature for a long time, i.e, the dual reflection of MOR signal between two reflectors. The dual reflection is a common phenomenon for reflectors having spatial correlation and has been widely considered in radar system \cite{Griesser1989,modi2019bounce,hua2019mr}. Specifically, the dual reflection happens between two reflectors in placement with dihedral angle such that beam lobes of two reflectors can point towards each other. Here, we exemplify the dual reflection between a pair of IRS in an indoor transmission scheme, as shown in Fig.~\ref{fig:newdemo}. 

Path \textbf{A} is the blocked path between a transceiver pair thus leveraging IRS is necessary. Paths \textbf{B} and \textbf{D} are line of sight (LoS) paths between \texttt{Tx} to an IRS, paths \textbf{E} and \textbf{F} are that of from one IRS to another IRS, and paths \textbf{C} and \textbf{G} are that of from one IRS to \texttt{Rx}. Then \texttt{Rx} can receive first-order-reflection (FOR) signal from a cascaded line of sight (C-LoS) path \textbf{B}-\textbf{C} and second-order-reflection (SOR) signal from a C-LoS path \textbf{D}-\textbf{F}-\textbf{G}. In addition, there are paths caused by dual reflection and without loss of generality, we introduce the dual reflection between IRS$_1$ and IRS$_2$. 
As an LoS path \textbf{E} exist between IRS$_1$ and IRS$_2$, the signal components impinging on IRS$_1$ and IRS$_2$ can be reflected towards each other due to side lobes. In particular, IRS$_2$ can receive the FOR signal from a C-LoS path \textbf{B}-\textbf{E}. Meanwhile, the second-order-reflection (SOR) signal via a C-LoS path \textbf{D}-\textbf{F}-\textbf{E} can be received by IRS$_1$ as well. In this case, the dual reflections between IRS$_1$ and IRS$_2$ are introduced. Immediately, the \texttt{Rx} further receive the SOR signal passing through a C-LoS path \textbf{B}-\textbf{E}-\textbf{G} and third-order-refection signal along another C-LoS path \textbf{D}-\textbf{F}-\textbf{E}-\textbf{C}. Due to the dual reflection will still occur, higher-order reflection signals are successively produced by repetitive signal reflections between IRS$_1$ and IRS$_2$ (for example, C-LoS paths \textbf{B}-\textbf{E}-\textbf{E}-\textbf{C} and \textbf{B}-\textbf{E}-\textbf{E}-\textbf{E}-\textbf{G}). As a result, some signal components keep continuous reflecting between the dual IRS pair, while other parts can either reach \texttt{Rx} or dissipate in trivial directions. Note that, signal components from higher-order reflections should not be neglected as long as they are not overwhelmed by \texttt{Rx}'s  noise power, or potential destruction of signal amplitude, fatal phase distortion and inter-symbol interference can significantly undermine the overall system performance. Thus, the dual reflection should be well considered in a complete signal model of IRS networks.

However, we notice most works about IRS only consider FOR. Though works \cite{Mei2021,mei2021intelligent} further consider C-LoS paths in MOR, the signal component via dual reflection is omitted in their models. To the best of the authors' knowledge, two main issues remained unsolved. First, no IRS works completely considered a complete channel model in the reflective environment. Thus, an establishment of the complete model for IRS network is necessary for analyzing generic and arbitrary reflecting scenarios. Note that, it is the most critical prerequisite to lay a foundation of a precise, robust, and reliable design for IRS network. Further, no analytical works have indicated clear bounds to guide the deployment of IRS networks with multi-user interference, i.e., how much EE and SE can be respectively improved, and where is the sum-rate upper bound and how to reach the upper bound. 

 \begin{figure}
\hspace*{0mm}
    \centering
    \adjincludegraphics[width=85mm,height=50mm,center]{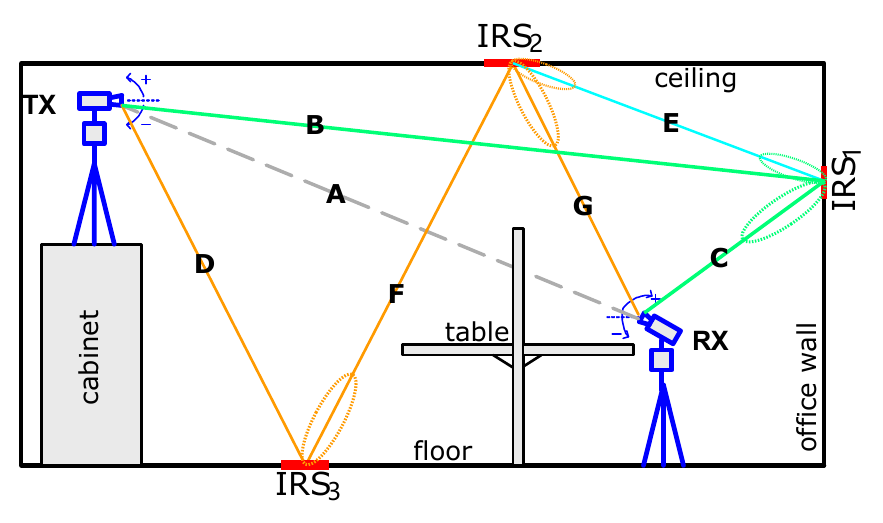}
    \caption{An example of indoor transmission assisted by IRS network with the same furniture setting shown in \cite{Maltsev2009}}
    \label{fig:newdemo}
\end{figure}

By addressing the above issues, the main contributions of this paper are listed as follows. 
 \begin{itemize}
     \item To incorporate the MOR effect with dual reflection, we introduce an index matrix to derive a complete model of IRS network, which is applicable for arbitrary orders of reflections, arbitrary number of IRS and arbitrary topologies of IRS network.  
     \item We mathematically derive two critical conditions: the optimal condition to reach the sum-rate upper bound and condition to realize interference-free transmission as insights for studying the EE and SE of IRS network.
    \item Considering different topologies of the IRS network, we analyze the sum-rate upper bound of MUMOR transmission assisted by an IRS network, by employing the optimal condition we derive and graph decomposition to realize the maximized EE and SE.
    
\end{itemize}

The rest of the paper is organized as follows. Section II derives two fundamental models of the MUMOR IRS network. In Section III, the MUMOR IRS network channel model is derived by permutationally combining fundamental models in Section II. Section IV derives the optimal condition and the interference-free transmission's condition. Section V obtains the sum-rate upper bound of the MUMOR IRS network in different network topologies. Simulations and conclusions are given in Section VI and Section VII, respectively. Proof of the optimal condition is given in Appendix A and proof of interference-free transmission's condition is shown in Appendix B.

\textit{Notations:} Throughout this paper, bold-faced upper case letters, bold-faced lower case letters, and light-faced lower case letters are used to denote matrices, column vectors, and scalar quantities, respectively. $\angle$ is the phase of a complex variable. The superscripts $(\cdot)^T$ and $(\cdot)^H$ represent matrix (vector) transpose, complex conjugate transpose, respectively. $\odot$ denotes point-wise multiplication. $\mathbf{I}$ is the identity matrix. The number of $Y$-combinations from a set $S$ of $X$ elements is denoted by ${X \choose Y}$, $\Perm{X}{Y}$ means the number of $Y$-permutations from a set $S$ of $X$ elements. $diag(\cdot)$ is the symbol for vectoring a matrix by taking its diagonal terms.


\section{Fundamental IRS Models }\label{sec:1}
In this section, two fundamental models in IRS networks are presented. We consider LoS channels obey the quasi-optical transmission nature of EM carrier following works \cite{tam1995raytrace,Cao2020,Bjornson2020,Perovic2019,tang2021pathloss,Maltsev2009}. Meanwhile, the NLoS channel between transceivers are considered, as no LoS paths between transceivers could be a common and a pressing issue \cite{bai2015los}, as shown in Fig.~\ref{fig:newdemo}. In addition, we assume each transceiver and IRS is located in a far-field as did in the literature \cite{Wu2019,xie2022gao,Zhou2020,Alexandropoulos2020,Sun2020,shaocheng2020,zijian2021,Kishk2021,Yu2020,mei2021mush,Mei2021,mei2021intelligent}. 
\subsection{The Single IRS Channel Model}
 As shown in Fig.~\ref{fig:demo}, we consider \textit{N} pairs of transceivers where each \texttt{Tx} or \texttt{Rx} is equipped with a single antenna, \textit{M} elements in ULA\footnote{Although ULA is adopted, the proposed IRS framework can be generalized to URA \cite{liu2022multi} or any other geometry.} 
 for each IRS piece and LoS channels between \texttt{Txs}/\texttt{Rxs} and IRS. Denote $\mathbf{A}_{in} \in \mathbb{C}^{M\times N}$ and $\mathbf{A}_{out} \in \mathbb{C}^{M\times N}$ as the LoS channel matrix of angle of arrivals (AOA) and angle of departures (AOD) form \texttt{Txs} to IRS and IRS to \texttt{Rxs}, respectively. Then, we have
\begin{equation} \label{eq:Amatrixin}
\mathbf{A}_{in}=[\mathbf{a}(\phi_{in,1}),\mathbf{a}(\phi_{in,2}),\dots,\mathbf{a}(\phi_{in,N})]\;
\end{equation}
and
\begin{equation} \label{eq:Amatrixout}
\mathbf{A}_{out}=[\mathbf{a}(\phi_{out,1}),\mathbf{a}(\phi_{out,2}),\dots,\mathbf{a}(\phi_{out,N})]\;,
\end{equation}
\noindent where $\mathbf{a}(\phi_{in,i})$ and $\mathbf{a}(\phi_{out,i})$ are steering vectors of incident directions $\phi_{in,i}$ and exit directions $\phi_{out,i}$ from \texttt{Tx}$_i$ to the IRS and IRS to \texttt{Rx}$_i$, respectively.
  \begin{figure}
\hspace*{0mm}
    \centering
    \adjincludegraphics[width=80mm,height=50mm,center]{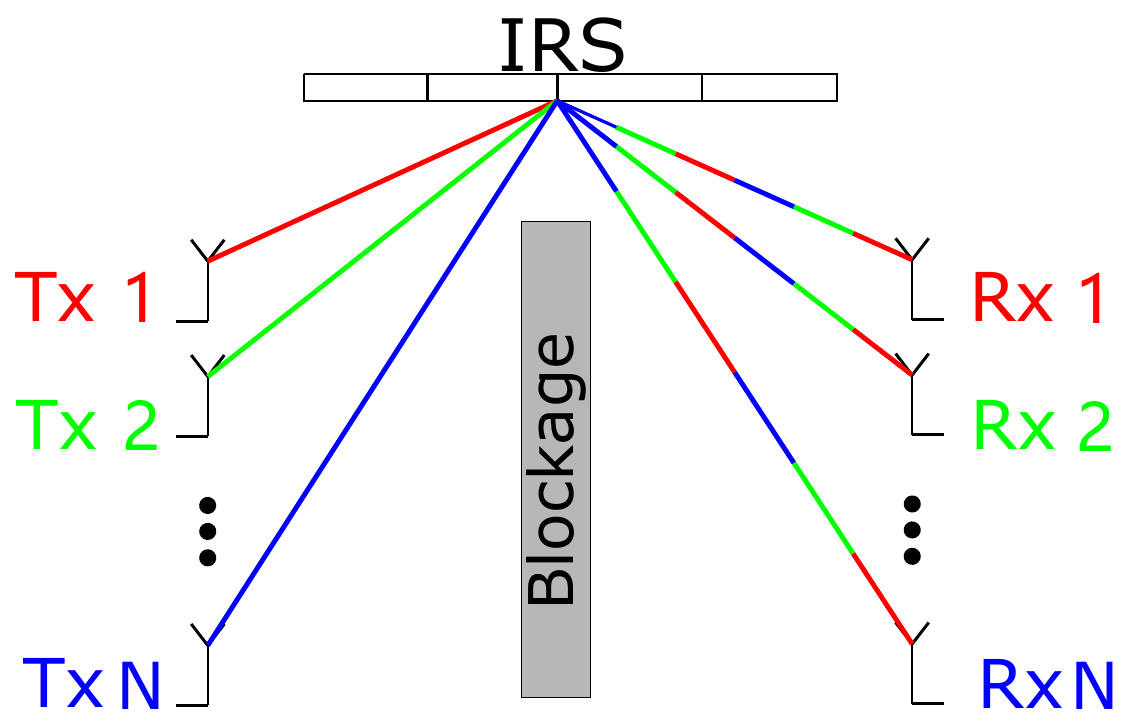}
    \caption{A single IRS model for MU transmission. Different colors mark the signal transmission path from different \texttt{Txs}.}
    \label{fig:demo}
\end{figure}
The IRS weights matrix $\mathbf{W} \in \mathbb{C}^{M\times M}$ is a diagonal matrix with each entity on the diagonal being the weight value. The received signal for all \texttt{Rxs} can be rewritten as 
\begin{equation}\label{eq:gener}
    {\mathbf{y}}=\mathbf{A}_{out}^T\mathbf{W}\mathbf{A}_{in}\mathbf{s}+\mathbf{n}\;,
\end{equation}
\noindent where $\mathbf{s}=[s_1,s_2,\dots,s_{N}]^T \in \mathbb{C}^{N \times 1}$ is the source signal vector from all \texttt{Txs}. In addition, $\mathbf{n}$ is the noise vector at the \texttt{Rxs}. The received signal of \texttt{Rx}$_{i}$ in Eq. (\ref{eq:gener}) can be rewritten as \cite{liu2022multi}
\begin{equation}\label{eq:gene_ith}
    {y}_{i}= \mathbf{w}^H\mathbf{A}_{C,i}\mathbf{s} + n_{i}\,, i=1, 2, ..., N\;\;,
\end{equation}
where $\mathbf{w}$ is a column vector whose elements are the main diagonal elements of $\mathbf{W}$. Meanwhile, $n_{i}$ is the noise at \texttt{Rx}$_{i}$. The $i$-th combined steering vector $\mathbf{A}_{C,i}$ can be written as 
\begin{equation}\label{eq:stemati}
   \mathbf{A}_{C,i}= [\mathbf{a}_{C}(\phi_{out,i}\,,\phi_{in,1}),\dots, \mathbf{a}_{C}(\phi_{out,i}\,,\phi_{in,N})]\in \mathbb{C}^{M\times N}\;,
\end{equation}
where 
\begin{multline} \label{eq:endest-1}
\mathbf{a}_{C}(\phi_{out,v}\,,\phi_{in,u}) ={l}_{IRS}\mathbf{a}(\phi_{out,v})\odot\mathbf{a}(\phi_{in,u}),\\ u,v=1,...\,,N, \
 \end{multline}
 and
 \begin{equation} \label{eq:endest-11}
 \mathbf{a}(\phi) =[1,e^{-jkd\cos\phi},\dots,\\
 e^{-jkd\cos\phi(M-1)}]^T\;.
 \end{equation}
Here, ${l}_{IRS}$ is the path-loss factor for LoS path and its specific expression has been given in \cite{tang2021pathloss,Ozdogan2020}. Without loss of generality, we assume the path-loss factor is a constant.

\subsection{Channel Model Between Two IRSs}
In this subsection, we derive the LoS channel model between one IRS to another as it is fundamental to make up a part of the complete model of IRS network. 

\begin{lemma}\label{lemma:tworis}
The channel matrix between any two IRSs is \textbf{rank-one} and can be written as \begin{equation}\label{eq:bigE}
\mathbf{E}=\mathbf{a}(\phi_{in})\mathbf{a}(\phi_{out})^T,
\end{equation}
where $\phi_{out}$ is the AOD of signal leave from the first IRS towards the next IRS, and $\phi_{in}$ is the AOA of signal arriving at the next IRS.
\end{lemma}
\begin{proof}
  We consider IRS$_A$ and IRS$_B$ have $M_A$, $M_B$ elements with element spacing $d_A$ and $d_B$ respectively. We denote $A_{i}$ and $B_{j}$ are the $i$-th element and $j$-th element on IRS$_A$ and IRS$_B$, $i \in [1,M_A], j \in [1,M_B]$. The relative distance from the $i$-th element on IRS$_A$ to the first element A$_1$ is $d_{A,i}$ and for that of IRS$_B$ is $d_{B,j}$ between the $j$-th element on IRS$_B$ and B$_1$.
  Since now we have two pieces IRS, to distinguish, we denote the azimuth AOD of IRS$_A$ between elements $A_{i}$ and $B_{j}$ as $\delta_{ij}$, and denote the azimuth AOA of IRS$_B$ as $\varepsilon_{ij}$. Then, we denote the distance between element $A_{i}$ on IRS$_A$ and element $B_{j}$ on IRS$_B$ as $D_{ij}$, and $\mu$ is the angle between IRS$_A$ and IRS$_B$, as shown in Fig.~\ref{fig:tworis}. We assume $D_{11}, \delta_{11}$ and $\varepsilon_{11}$ is known, and there is $\mu=\varepsilon_{11}-\delta_{11}$. From the trigonometric relationship, we have
\begin{equation}\label{eq:epsij}
    \varepsilon_{ij}=\tan^{-1}\left(\frac{D_{11}\sin\varepsilon_{11}-d_{A,i}\sin\mu}{D_{11}\cos\varepsilon_{11}-d_{A,i}\cos\mu-d_{B,j}}\right).
\end{equation}
The distance $D_{ij}$ between elements $A_{i}$ and $B_{j}$ can be calculated as
\begin{equation}\label{eq:dij}
    D_{ij}=\frac{D_{11}\sin\varepsilon_{11}-d_{A,i}\sin\mu}{\sin\varepsilon_{ij}}\;.
\end{equation}
Since the far-field condition holds where distance is much greater than the aperture of IRS such that D$_{11}>> Md$, we have $\varepsilon_{ij}\approx \varepsilon_{11}$ and $\delta_{ij}\approx \delta_{11}$ correspondingly. Thus, by substituting $\varepsilon_{ij}$ with $\varepsilon_{11}$ in Eq. (\ref{eq:epsij}), we have
\begin{equation}\label{eq:daidbi}
  d_{A,i}\sin\delta_{11}=-d_{B,i}\sin\varepsilon_{11}.
\end{equation}
Then, we substitute Eq. (\ref{eq:daidbi}) into Eq. (\ref{eq:dij}), we have
\begin{equation}
    D_{ij}=D_{11}-d_{A,i}\cos\delta_{11}-d_{B,j}\cos\varepsilon_{11}.
\end{equation}
Since the LoS channel between IRS$_A$ and IRS$_B$ can be represented by
\begin{equation}
\mathbf{E}_{AB}=e^{-jk\mathbf{D}},
\end{equation}
where $\mathbf{D} \in \mathbb{C}^{M_A \times M_B}$ is the distance matrix derived from $D_{ij}$, $\mathbf{E}_{AB}$ can be rewritten as 
\begin{equation}\label{eq:tworis}
\mathbf{E}_{AB}=e^{-jkD_{11}}(\mathbf{a}(\varepsilon_{11})\mathbf{a}(\delta_{11})^T)^*.
\end{equation}
\noindent Note that the path delay $D_{11}$ is a constant between any two fixed IRSs. Since the path delay is known, it can be removed here. In addition, by taking the inverse element order of IRS$_A$ and IRS$_B$, which is equal to taking conjugate to the steering vectors of AOA and AOD, Eq. (\ref{eq:tworis}) can be rewritten as 
\begin{equation}\label{eq:tworis_conju}
\mathbf{E}_{AB}=\mathbf{a}(\varepsilon_{11})\mathbf{a}(\delta_{11})^T.
\end{equation}
As $\delta_{11}=\phi_{out}$, $\varepsilon_{11}=\phi_{in}$, we can observe that the LoS channel between arbitrary two IRSs can be considered as the out product of two steering vectors, which is a rank one matrix given in Eq. (\ref{eq:bigE}).
\end{proof}
 From the view of Lemma \ref{lemma:tworis}, each IRS can regard another IRS as a point source located in the far-field, but both of them are able to shape a pencil beam towards each other. Additionally, it means that only a single data stream can be supported by a LoS channel between two pieces of IRS. In real applications, ranks can be greater than one due to diffraction and refraction effects of EM wave. However, as the carrier frequency keeps increasing for more spectrum resource, the diffraction and refraction effects becomes weak and vulnerable. Hence multiple streams transmission between two IRSs becomes impractical, where traditional rayleigh fading model is inconsistent in this case \cite{Bjornson2020}. Therefore, in this work, we consider rank-one channel between two IRSs.
 
 \begin{figure}
   \centering
  \includegraphics[width=60mm,height=60mm]{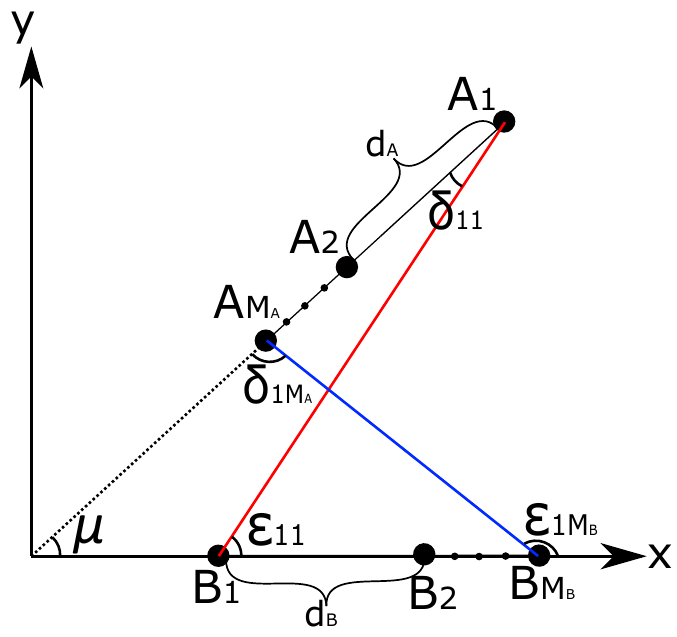}
  \caption{The illustration of channel model between two IRSs.}
  \label{fig:tworis}
\end{figure}
 

\section{IRS network Channel Model}

 \subsection{The MUFOR Network Channel}
 Denote the channel of FOR IRS network as $\mathbf{H}_{I,1}$. Based on Eq. (\ref{eq:gener}), the received signal with $K$ pieces IRS can be expressed as 
\begin{equation}\label{eq:macro_MUMOR}
    {\mathbf{y}}=\mathbf{H}_{I,1}\mathbf{s}+\mathbf{n}\;,
\end{equation}
where 
\begin{equation}\label{eq:FOR-mu-model}
\mathbf{H}_{I,1}=\sum_{k=1}^{K}\mathbf{A}_{out,k}^T\mathbf{W}_k\mathbf{A}_{in,k}\;
\end{equation}
and $\mathbf{A}_{in,k}$ and $\mathbf{A}_{out,k}$ are two steering vector matrices of AOA and AOD with respect to $k$-th IRS. Note that, the $k$-th C-LoS path component is made up via multiplexing only one weights matrix $\mathbf{W}_k$ one time with other two steering vector matrices, $\mathbf{A}_{in,k}$ and $\mathbf{A}_{out,k}$. Thus, the IRS network channel $\mathbf{H}_{I,1}$ embody $K$ different paths and all of these paths only experience one time reflection.
\subsection{The Exemplification of MUMOR Network Channel}
To model the MOR effect analytically, we define the maximum order of reflections which can exist within the IRS network as $\Gamma$. Essentially, $\Gamma$ plays a role of effective cut-off parameter on the MOR effect. Though it is possible to consider $\Gamma \rightarrow \infty$\footnote{It is similar to the LoS path of visible light reflected within two mirrors or more mirrors for infinite times.}, we need to cut off using $\Gamma$ as a finite value because we have path-loss in practical scenarios. In this case, the reflection order of signal components less than or equal to $\Gamma$ is considered, while the signal components with orders higher than $\Gamma$ are assumed to be overwhelmed by the noise power and hence can be neglected.
\begin{figure*}
   \centering
   \adjincludegraphics[width=\textwidth,height=35mm,right]{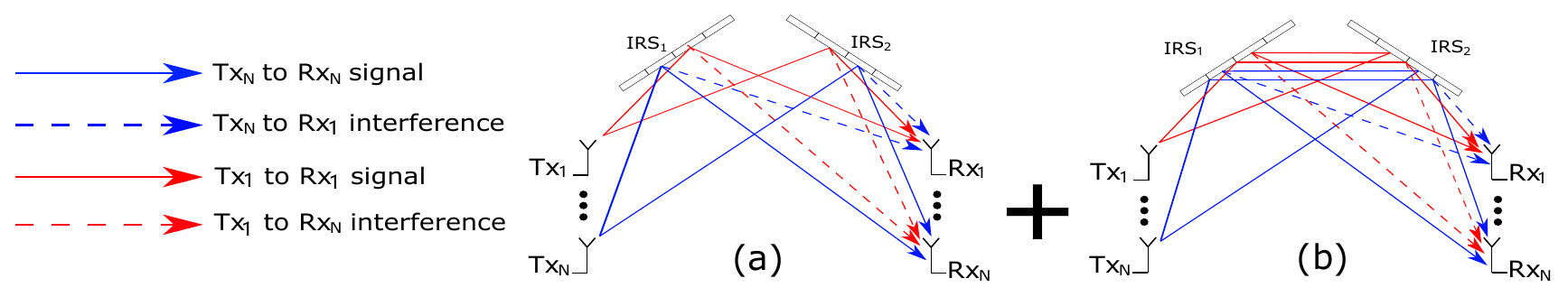}
  \caption{MUMOR transmission within the IRS network given $\Gamma=2$, $K=2$. (a) MU signals passing along the FOR IRS network channel, $\mathbf{H}_{I,1}$. (b) MU signals passing along the SOR IRS network channel, $\mathbf{H}_{I,2}$.}
  \label{fig:nris}
\end{figure*}
To involve arbitrary number of reflections in the IRS network, we denote the IRS network channel in $\gamma$-th order as $\mathbf{H}_{I,\gamma},$ where $\gamma \in [1,\Gamma]$. Then, we extend the IRS network channel in Eq. (\ref{eq:FOR-mu-model}) from FOR to MOR via superposition as  
\begin{equation}\label{eq:p_sum}
        \mathbf{H}_I=\sum_{\gamma=1}^{\Gamma}\mathbf{H}_{I,\gamma}\;,
\end{equation}
where totally $\Gamma$ orders of IRS network channels are added up. The $\gamma$-th order MU channel component $\mathbf{H}_{I,\gamma}$ includes all C-LoS path components that experience $\gamma$ orders in the network, which also means each C-LoS path component in $\mathbf{H}_{I,\gamma}$ is exactly weighted for $\gamma$ times. 

To differentiate MOR from FOR in the IRS network and illustrate the dual reflection, we exemplify by considering two pieces of IRS, where $K=2$ and $\Gamma=2$, as shown in Fig.~\ref{fig:nris}. In this case, each C-LoS path passes maximal $2$ pieces IRS. The C-LoS paths with a reflection order of more than three are ignored. Thus, the received signal of all receivers should consist of MU signals passing along the FOR IRS network channel $\mathbf{H}_{I,1}$ and the SOR IRS network channel $\mathbf{H}_{I,2}$ which has been respectively shown in Fig.~\ref{fig:nris}(a) and (b). Considering the FOR IRS network channel, we can have
\begin{equation}\label{eq:p0}
\mathbf{H}_{I,1}=\mathbf{A}_{out,1}^T\mathbf{W}_{1}\mathbf{A}_{in,1}+\mathbf{A}_{out,2}^T\mathbf{W}_{2}\mathbf{A}_{in,2}  \;\;.
\end{equation}
For SOR IRS network channel, we have
\begin{equation}\label{eq:p2}
\mathbf{H}_{I,2}=\mathbf{A}_{out,2}^T\mathbf{W}_{2}\mathbf{E}_{12}
    \mathbf{W}_{1}\mathbf{A}_{in,1}+\mathbf{A}_{out,1}^T\mathbf{W}_{1}\mathbf{E}_{21}
    \mathbf{W}_{2}\mathbf{A}_{in,2}  \;\;.
\end{equation}
Note that $\mathbf{E}_{12}$ and $\mathbf{E}_{21}$ are the LoS channels between between IRS$_1$ and IRS$_2$, as we derived in Eq. (\ref{eq:tworis_conju}), where $\mathbf{E}_{12}=\mathbf{E}_{21}^T$.

We can observe the number of C-LoS path components in $\mathbf{H}_{I,\gamma}$ with different $K$ and different $\gamma$ variates and still follows the permutation's rule. For example, for $K=2$ we have the IRS candidate set $\kappa=\{1,2\}$ which means there are only IRS$_1$ and IRS$_2$ in the environment. Since $\gamma=1$, the number of FOR paths is equal to $2$ as one FOR path passes through IRS$_{1}$ and another one passes through IRS$_{2}$. Using the permutation rule, we can denote the number of FOR paths as $\Perm{2}{1}=2$ ($\Perm{X}{Y}$ means $Y$-permutations of a set with $X$ elements, where $X,Y \in \mathbb{N^+}$). Similarly, for $K=2$ and $\Gamma=2$, the number of SOR paths equal to $2$ since $\Perm{2}{2}=2$. Specific order sequences of these two SOR paths can be enumerated here, i.e., we have $[1\;\;2]$, meaning a SOR path first passes through IRS$_1$ and then IRS$_2$, and $[2\;\;1]$, meaning another SOR path passes through IRS$_2$ and then IRS$_1$. Consequently, the total number of C-LoS paths of $\gamma$ orders in $K$-piece IRS network is equal to $\Perm{K}{\gamma}$. Note that although $\Perm{K}{\gamma}$ only includes the number of C-LoS paths which pass each IRS only once in IRS networks, we will discuss and consider C-LoS paths which repetitively visit a same IRS later.

To expand $\mathbf{H}_{I,\gamma}$ in general expression, we define an index matrix $\mathbf{X}_{\gamma}$ to denote the order sequences for all C-LoS paths in $\gamma$ orders. In particular, all rows of index matrix $\mathbf{X}_{\gamma}$ are used to hold specific order sequences of all C-LoS paths of $\gamma$ orders. For example, given $K=2, \gamma=2$, by leveraging the index matrix $\mathbf{X}_{2}$, Eq. (\ref{eq:p2}) can now be written as
\begin{equation}\label{p22}
\mathbf{H}_{I,2}=\sum_{u=1}^{\Perm{2}{2}}
    \mathbf{A}_{out,{X}_{\gamma,u2}}\mathbf{W}_{X_{\gamma,u2}}\mathbf{E}_{X_{\gamma,u1}X_{\gamma,u2}}\mathbf{W}_{X_{\gamma,u1}}\mathbf{A}_{in,{X}_{\gamma,u1}}\;,
\end{equation}
\noindent where the index matrix $ \mathbf{X}_{2}$ for $\mathbf{H}_{I,2}$ is
\begin{equation}
    \mathbf{X}_{2}=\begin{bmatrix}
  X_{2,11} & X_{2,12} \\
  X_{2,21} & X_{2,22} \\
    \end{bmatrix}=\begin{bmatrix}
  1 & 2 \\
  2 & 1 \\
    \end{bmatrix}.
\end{equation}
We can observe $[X_{2,11}\;\;X_{2,12}]$=[1\;\;2] and $[X_{2,21}\;\;X_{2,22}]=[2\;\;1]$ are exactly two sequences we enumerate. 
\subsection{The MUMOR Network Channel}
 
For arbitrary value of $\gamma$ and $K$, we define the index matrix as $\mathbf{X}_{\gamma}\in \mathbb{N^+}^{\Perm{K}{\gamma} \times \gamma}$. The term ${X_{\gamma,uv}}$ at the $u$-th row and the $v$-th column of $\mathbf{X}_{\gamma}$ is a positive integer representing an index of a specific IRS in the network. Each row of $\mathbf{X}_{\gamma}$ holds a specific and non-repetitive sequence with $\gamma$ columns. The index matrix $\mathbf{X}_{\gamma}$ has $\Perm{K}{\gamma}$ rows in total, which means all order sequences under a partial permutation $\Perm{K}{\gamma}$ are included. To generate the index matrix, one can enumerate the permutation sequences of $\gamma$ terms from the IRS candidate set $\kappa=\{1, 2, ..., K \}$ respectively as rows of $\mathbf{X}_{\gamma}$ \cite{ord1970generation}. 



\begin{theorem}\label{theo:netchan}
The general expression of $\gamma$-order IRS network channel $\mathbf{H}_{I,\gamma}$ can be written as
\begin{multline}\label{eq:pgamma}
\mathbf{H}_{I,\gamma}=\sum_{u=1}^{\Perm{K}{\gamma}}
    \mathbf{A}_{out,{X}_{\gamma,u\gamma}}
    [\prod_{v=1}^{\gamma-1}\mathbf{W}_{X_{\gamma,u(v+1)}}\mathbf{E}_{X_{\gamma,uv}X_{\gamma,u(v+1)}}]...\\
    \mathbf{W}_{X_{\gamma,u1}}\mathbf{A}_{in,{X}_{\gamma,u1}}\,
\end{multline}
where $\mathbf{X}_{\gamma}\in \mathbb{N^+}^{\Perm{K}{\gamma} \times \gamma}$ is the index matrix of $\gamma$ orders. 
\end{theorem}
Note that the dual reflection is common, and we should consider other C-LoS paths whose order sequences are with repetitive indices. These C-LoS paths should at least visit a single IRS of all pieces twice. For the order sequences of these C-LoS paths with repetition, the adjacent two terms in rows of the index matrix should be different as we consider that there is no LoS path between one IRS and itself thus the EM wave would not impinge on the same IRS twice immediately, i.e., $X_{\gamma,uv} \neq X_{\gamma,u(v+1)}, u \in [1,\Perm{K}{\gamma}],v \in [1,\gamma-1]$. To complete the IRS network model, we include order sequences, whose two interleaved indices can be equal to another, into the index matrix $\mathbf{X}_{\gamma}$ with extra rows. Therefore, the row dimension of $\mathbf{X}_{\gamma}$ extends from ${\Perm{K}{\gamma}}$ to ${K(K-1)^{(\gamma-1)}}$. By far, if we replace $\mathbf{H}_{I,1}$ with Eq. (\ref{eq:p_sum}) and Eq. (\ref{eq:pgamma}) in Eq. (\ref{eq:macro_MUMOR}), then the complete model of IRS network is established. 

\section{Proposed Theorems For A Single IRS}
\subsection{Optimal Sum-rate Condition Based On Single IRS}
In case of single IRS, where $\Gamma=1$, $K=1$ in Eq. (\ref{eq:p_sum}), then the received signal for all \texttt{Rxs} becomes:
\begin{equation}\label{eq:ideal_rf}
    \mathbf{y}=\mathbf{H}_{I,1}\mathbf{s}+\mathbf{n};=\mathbf{A}_{out,1}^T\mathbf{W}\mathbf{A}_{in,1}\mathbf{s}+\mathbf{n};,
\end{equation}
where
\begin{equation} \label{eq:General channel}
\mathbf{H}_{I,1}=\\
 \begin{bmatrix} 
 \mathbf{w}^{H}\mathbf{a}_{C}(\phi_{out,1}\,,\phi_{in,1}) & \!\!\!\dots\!\!\!  & \mathbf{w}^{H}\mathbf{a}_{C}(\phi_{out,1}\,,\phi_{in,N}) \\
 \mathbf{w}^{H}\mathbf{a}_{C}(\phi_{out,2}\,,\phi_{in,1}) & \!\!\!\dots\!\!\!  & \mathbf{w}^{H}\mathbf{a}_{C}(\phi_{out,2}\,,\phi_{in,N}) \\
 \vdots   & \!\!\!\ddots\!\!\!& \vdots     \\
\mathbf{w}^{H}\mathbf{a}_{C}(\phi_{out,N}\,,\phi_{in,1}) & \!\!\!\dots\!\!\! &  \mathbf{w}^{H}\mathbf{a}_{C}(\phi_{out,N}\,,\phi_{in,N})\\
 \end{bmatrix} .
\end{equation}

With equal power $P_T$ from all transmitters, the channel capacity of MU transmission on single IRS can be expressed as 
\begin{equation}\label{eq:cap_sin}
      C= log\det(\mathbf{I}_{N}+\frac{P_T}{N_{0}}\mathbf{H}_{I,1}\mathbf{H}_{I,1}^H)\;.
\end{equation}
Then, the optimization on weights $\mathbf{w}$ is equivalent to maximize the diagonal terms and minimize the off-diagonal terms in Eq. (\ref{eq:General channel}). However, it is hard to decide whether the main diagonal terms and the off diagonal terms of $\mathbf{H}_{I,1}$ can be simultaneously maximized and minimized, i.e, $|\mathbf{w}^{H}\mathbf{a}_{C}(\phi_{out,i}\,,\phi_{in,i})|=M$ and $|\mathbf{w}^{H}\mathbf{a}_{C}(\phi_{out,i}\,,\phi_{in,j})|=0,\;i\neq j$.

Note that once the spatial correlation between each transceiver pairs $\mathbf{a}_{C}(\phi_{out,u}\,,\phi_{in,v}),\;u,v=1,...,N$ are fixed, we can calculate $\mathbf{w}$. As the IRS channel is deterministic with fixed $\mathbf{w}$, the spatial correlation between all transceiver pairs is another dominating factor for deciding the sum-rate upper bound. For example, the higher the spatial channel between \texttt{Tx}$_{i}$ and \texttt{Tx}$_j$ or between \texttt{Rx}$_{i}$ and \texttt{Rx}$_j$, the lower the channel ranks and singular values of $\mathbf{H}_{I,1}$ are, which further lower the upper bound of the overall sum-rate in a specific spatial realization. To find an optimal upper bound of sum-rate, we derive the optimal condition in spatial correlation between each transceiver pair. In this case, every pair can leverage the optimal gain brought by the single IRS. Besides, the interference between each pair can be nullified simultaneously.

Let the $i$-th pair user locate at $\phi_{in,i}=\alpha_{i}\,,\phi_{out,i}=\beta_{i}$ and the $j$-th pair locate at $\phi_{in,j}=\alpha_{j}\,,\phi_{out,j}=\beta_{j}$ where $i \neq j, i,j=1,2,...,N$. Denote $\Delta r=\frac{d}{\lambda}$ as the normalized spacing between each element since $d$ is the distance between each element and $\lambda$ is the carrier wavelength. Additionally, denote $L=M\Delta r$ is the relative length respect to normalized spacing. Then we have  

\begin{equation}
w_{m}=e^{-j\zeta_{m}kdm}\,,\theta_{m}\in(0,2\pi]\,,m=1,\dots,M\;\;,
\end{equation}
where
\begin{equation}
    \zeta_m=-\cos\alpha_{i}-\cos\beta_{i}+\frac{K}{\Delta r}
\end{equation}
is the optimal factor given by maximal ratio combining (MRC) algorithm to realize power gain of the $i$-th pair user, i.e., $|\mathbf{w}^{H}\mathbf{a}_{C}(\phi_{out,i}\,,\phi_{in,i})|=M$.

\begin{lemma}\label{theo:opp}
Given $N$ pairs of transceivers assisted by a single piece IRS, the optimal sum-rate upper bound can be obtained when each transceiver pair's position for \texttt{Tx} and \texttt{Rx} are at
%
\begin{equation}\label{eq:optimal pos_ula_a}
\alpha_{j}=\cos ^{-1}\left(\frac{j}{L}-\zeta_m-\cos \beta_{i} \pm \frac{1}{\Delta r}\right)
\end{equation}
and
 \begin{equation}\label{eq:optimal pos_ula_b}
\beta_{j}=\cos ^{-1}\left(\frac{j}{L}-\zeta_m-\cos \alpha_{i} \pm \frac{1}{\Delta r}\right)
\end{equation}
 respectively.
\end{lemma}
Lemma \ref{theo:opp} reveals that if the position of each transceiver can be coordinated correspondingly, $\mathbf{H}_{I,1}$ can be optimized such that diagonal terms can be maximized and off-diagonal terms can be nullified respectively at the same time. Physically, each reflected beam towards each \texttt{Rx} is orthogonal to each other. The whole IRS channel can be orthogonal space-division multiplexed (OSDM) by $M$ pairs of the transceiver. Thus, we call links with correlation obeying Lemma \ref{theo:opp} as the optimal link. Proof of Lemma \ref{theo:opp} is shown in Appendix A. 

Additionally, without changing $L$, no matter how the element number and spacing variate, the optimal condition in Lemma \ref{theo:opp} will not change. In fact, the characteristic of channel rank is essentially proportional to $L$ \cite{tse2005fundamentals}. Therefore, we propose to use $d=\frac{\lambda}{2}$, since this is the maximal spacing for a fixed $L$ to secure a narrowest reflected beam, which causes no grating lobe of the reflected beam. Note that, there is a trade-off between energy efficiency and spacing as well. This is due to smaller spacing resulting in fewer channel ranks and larger beamwidth. Still, the redundant beam, causing energy waste in trivial directions, will less likely occur \cite{liu2022multi}. Thus, the actual spacing can be less than this value based on different design criteria. Some discussions about the spacing of elements and the beamwidth of IRSs can be referred in \cite{Ozdogan2020,Liu2019,Han2021}.  
With fixed $L$, $M$, half-wavelength spacing, we have
\begin{theorem}\label{theo:3}
Given all transceivers are optimally positioned, the upper bound of the sum-rate for a single IRS is 

\begin{equation}\label{eq:limit_single}
      C_{Max}= Nlog(1+\frac{P_{T}M^2}{N_{0}}),\;\textnormal{if}\; N\leq M
           \;,
\end{equation}
where $P_{T}$ is the power of \texttt{Txs}, $N$ is the spatial multiplexing gain and $N_{0}$ is the noise power at the \texttt{Rxs}, given $M$ pairs. 
\end{theorem}
Based on Theorem \ref{theo:3}, the sum-rate upper bound is reached when $N=M$ and each pair receives the power gain of $M^2$. If normalized power from the \texttt{Tx} is considered without path-loss, then the power gain should be $1$ since the whole IRS network is a passive system.
When $N>M$, the interference between users is unavoidable, and now sum-rate should be determined specifically by the spatial correlation of transceivers and the ratio between $N$ and $M$. Other multiplexing schemes are proposed to avoid the inter-user interference if $N>M$. However, this case is rare in the real situation since $N\leq M$ can be guaranteed as there can be hundreds of thousands of IRS elements while keeping the far-field condition \cite{emil2020nearf}. 

Note that, the upper bound in Theorem \ref{theo:3} is hard to achieve as transceivers can not always stay at the optimal position in Lemma \ref{theo:opp}. Moreover, the element spacing may be less than half wavelength, and the mutual coupling effect can be an issue \cite{Bjornson2021coup}. Nevertheless, Theorem \ref{theo:3} is meaningful as it analytically provides a sum-rate upper bound for each IRS and can only be obtained by satisfying the optimal condition in Lemma \ref{theo:opp}.
\subsection{Interference-free Condition Based on A Single IRS}

When spatial correlation between transceiver pairs are not orthogonal, with one IRS, we can nullify the interference to achieve interference-free transmission. 



        \begin{lemma}\label{lemma:inff}
        In order to achieve interference-free transmission without orthogonal spatial correlations between each transceiver pairs, the element number on a single IRS should satisfy $M \geq N^2$.
        \end{lemma}
\noindent The interference-free condition can be easily achieved in practical deployment since each IRS can have sufficient amount of elements. The proof is given Appendix B, where we also show a single IRS is able to support multiple streams transmission with only one vector.

As there is a similarity and equivalence in the function between IRS and MIMO precoding/decoding, with the deployment of the IRS, transceivers can transfer some workloads to the IRS. Thus the structure of transceivers can be simplified.  Nevertheless, compared with traditional scheme, IRSs still have its unique advantages over the traditional MU scheme. In particular, the IRS can suppress the inter-user interference before receivers are jammed while the traditional scheme can not conveniently suppress the inter-user interference at receivers due to the joint decoding is usually not available.

         
        
\begin{remark}
 The number of transceivers that access the IRS network from a single IRS should be significantly below the number of elements of that a single IRS.
 And it is better transceivers can locate in a much more different direction than one IRS, or extra pieces IRS nearby should be involved to solve this issue since extra pieces IRS can distinguish these transceivers from a much more different location.
\end{remark}

 \begin{figure}
  \includegraphics[width=90mm,height=70mm,center]{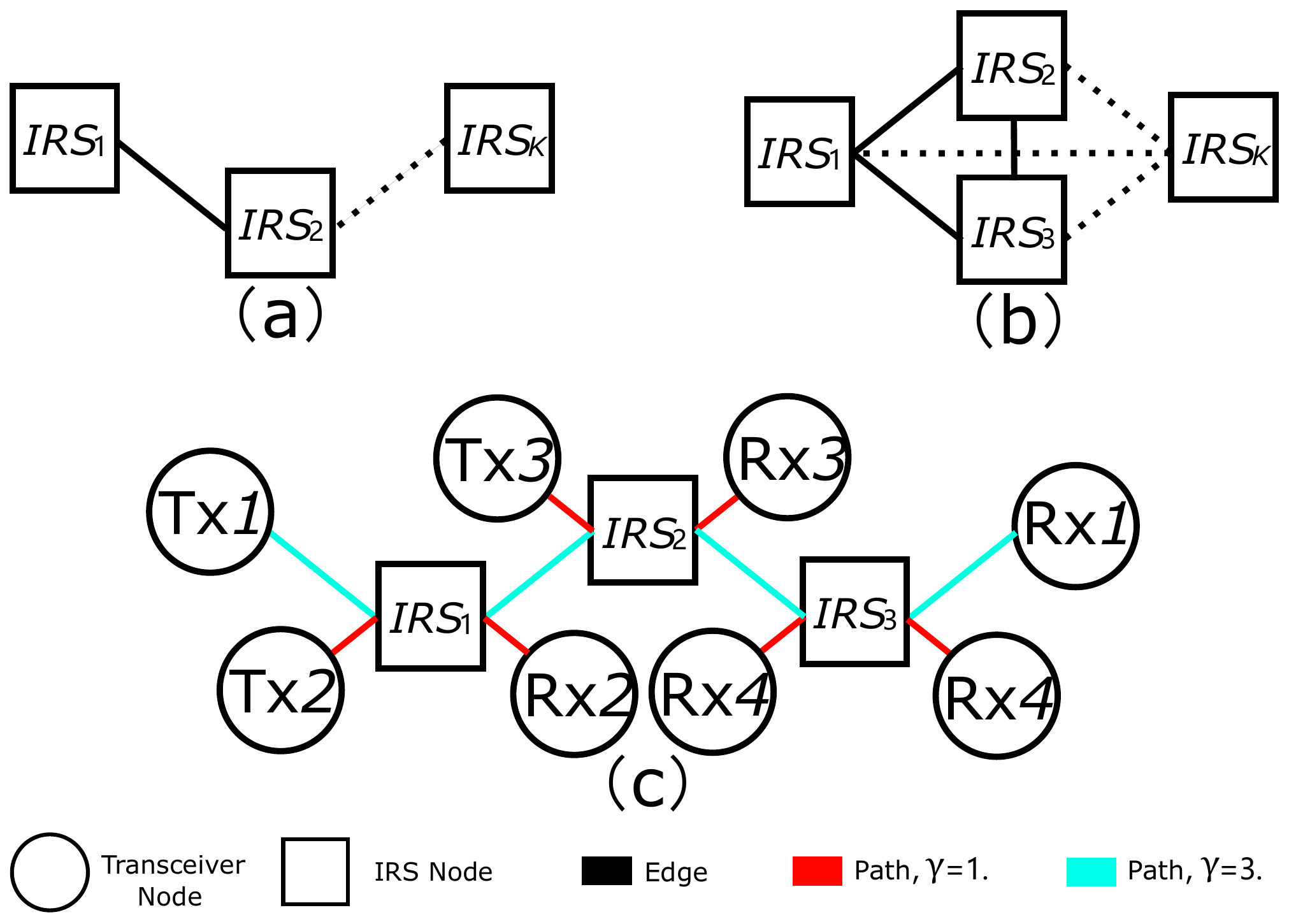}
  \caption{MUMOR Transmission based on IRS network with $K$ single IRS. (a) An LG topology. (b) A CG topology. (c) An example of a network shaping an LG to serve MU where $K$=3, $M=2$, $\Gamma=3$, $N$=4. The solid line represents the edge that connects two adjacent nodes. The dashed line represents a series of other adjacent connections that are omitted.}
  \label{fig:Multi_two}
\end{figure} 

\section{Analysis on Sum-rate upper bound of IRS Network}
For simplicity, we use the terminology of graph theory for the following discussion \cite{Jain2005}. We call an LoS channel as an edge, a single IRS/transceiver as a node, nodes connected to one node by an edge as adjacent nodes, the number of edges that are incident to a node as the degree, a C-LoS path as a path, the number of IRS nodes that the path passes through as the path length or simply length, and the IRS network as the network.

In the network, the sum-rate is affected by the network's topology and geometry, number of IRS/transceiver nodes and the weights design of IRS nodes. The topology is the connection statement of nodes by edges existing within the network while the geometry is determined by the relative AOA and AOD between arbitrary two nodes. Thus all nodes and edges have specific topological and geometrical relationship between each other, as shown in Theorem \ref{theo:netchan}. However, it is difficult to derive the exact sum-rate upper bound without prior determining to the network topology, geometry, and weights design.

Note that, Theorem \ref{theo:3} indicates each IRS node can fulfill criteria of power maximization and interference nullification given the optimal condition in Lemma \ref{theo:opp}. To maximize the EE and SE performance, we leverage Lemma \ref{theo:opp} to determine the network's geometry, topology and weights design. In particular, each IRS node can optimally serve other adjacent IRS/transceiver nodes, where maximally $M$ pair of adjacent nodes can be supported, or $2M$ degrees can be possessed by one IRS node. As the topology is versatile given $K$ nodes of IRS to form a network, we derive the sum-rate upper bound for two kinds of common graphs, which are linear graph (LG)\footnote{\textbf{Linear Graph/Path Graph}: a linear graph is a graph whose vertices/nodes can be listed in the order $v_1, v_2,..., v_n$ such that the edges exist between $v_i$ and $v_{i+1}$ where $i=1,2,...,N$. Paths are often important in their role as subgraphs of other graphs, in which case they are called paths in that graph.} and complete graph (CG)\footnote{\textbf{Complete Graph}: A complete graph is one in which every two vertices/nodes are adjacent: all edges that could exist are present.}, as shown in Fig.~\ref{fig:Multi_two}(a) and Fig.~\ref{fig:Multi_two}(b) respectively. In addition, a special topology without edges are considered, which is called the null graph (NG) and means each IRS only form FOR paths locally without LoS between any two IRS nodes in the network.

\subsection{The IRS Network In Linear Graph}
With the signal of \texttt{Tx}$_{i}$ passing along an LG with length of $K$ order, where $\Gamma=K$, we can write the received signal of \texttt{Rx}$_{i}$ after $K$ orders reflection as
\begin{multline}\label{eq:yi_gamma_mat}
     {y}_{i}=\mathbf{a}(\phi_{out,i,{X}_{K,1K}})
    [\prod_{v=1}^{K-1}\mathbf{W}_{{X}_{K,1(v+1)}}\mathbf{E}_{{X}_{K,1v}{X}_{K,1(v+1)}}]...\\
    \mathbf{W}_{X_{K,11}}\mathbf{a}(\phi_{in,i,{X}_{K,11}})s_i+n_i,
\end{multline}
where $\mathbf{a}(\phi_{out,i,{X}_{K,1K}})$ and $\mathbf{a}(\phi_{in,i,{X}_{K,11}})$ are the corresponding steering vector in matrix $\mathbf{A}_{out,{X}_{K,1K}}$ and $\mathbf{A}_{in,{X}_{K,11}}$ for $i$-th pair transceiver. Since each transceiver pair now communicates orthogonally in the network following optimal condition, the index matrix $\mathbf{X}_{K}$ now is simplified to contain only one row, holding one specific sequence of one C-LoS path. Note that though the dual reflection exist within the LG network as well. As $\Gamma=K$, there is only one path with the maximal effective length that can reach to \texttt{Rx}$_{i}$. Moreover, Eq. (\ref{eq:yi_gamma_mat}) can be written in a similar form with Eq. (\ref{eq:gene_ith}) such that
\begin{equation}\label{eq:LG}
    {y}_{i}=[\prod_{v=1}^{K}\mathbf{w}_{X_{1v}}^H\mathbf{a}_{C,i,{X}_{1v}}]s_i+n_i,
\end{equation}
where $\mathbf{a}_{C,i,k}$ means the equivalent channel of $IRS_{k}$ for $i$-th pair transceiver and $\mathbf{w}_{k}=diag(\mathbf{W}_{k})$ is the corresponding weights vector on $IRS_{k}$ for $k=1,...,K$.
 
As the optimal power gain for a single pair transceiver is $M^2$ from a single IRS, with $K$ order reflection where each IRS applying weights to realize maximal power gain in Eq. (\ref{eq:LG}), the cascaded power gain would be $M^{2K}$. In this case, EE is maximized for a single pair in an LG network. Thus, base on Eq. (\ref{eq:LG}), the sum-rate upper bound for one transceiver pair is 
\begin{equation}\label{eq:c_lg}
   C_{SU,LG,(K)}=log(1+\frac{P_{T}M^{2K}}{N_{0}})\;, 
\end{equation}
where the subscripts $SU,LG$, and $(K)$ mean a single pair, linear graph and a $K$-order reflection, respectively. Since an edge between two nodes is a rank one channel from Lemma 1, we are unable to realize multi-stream information transmission based on one edge and thus the cascaded channel of an LG is rank one. 

Nevertheless, MU transmission in an LG network is still available as each IRS nodes can have $2M$ degrees. E.g., with topology of network as shown in Fig. \ref{fig:Multi_two}(c), the sum rate upper bound can be reached by combining three $1$-length paths and a $3$-length path where each IRS node has $4$ degrees. By including the sum-rates from all $1$-length paths and one $K$-length path,
we have the MU sum-rate upper bound as 
\begin{equation}\label{eq:chain}
   C_{MU,LG,(1,K)}=log(1+\frac{P_{T}M^{2K}}{N_{0}})+ K(M-1)log(1+\frac{P_{T}M^{2}}{N_{0}}),
\end{equation}
where power of all \texttt{Txs} is equal to $P_T$ and the subscript $(1,K)$ means only the paths whose lengths are equal to $1$ and $K$ are involved. In this case, spatial multiplexing has been maximized while these paths would not introduce extra interference from reflections or dual reflections since all paths still keep spatially orthogonal.

\subsection{The IRS Network In Complete Graph}
For CG network, though multiple paths can be leveraged by one pair transceiver, this is equivalent to transfer spatial multiplexing into power gain which introduces a trade-off. To maximize spatial multiplexing gain of network, each transceiver should send one stream via one path. Thus, the network sum-rate depends on how many Eulerian paths\footnote{\textbf{Eulerian path}: or Eulerian trail is a trail in a finite graph that visits every edge exactly once (allowing for revisiting vertices/nodes).} without revisiting nodes in the CG network. Eulerian paths with revisiting nodes are excluded due to these transmissions are not necessary in the network.
\begin{figure*}
       \centering
  \subcaptionbox{$M=64$, $L=2$.\label{fig3:a}}{\includegraphics[width=2.375in]{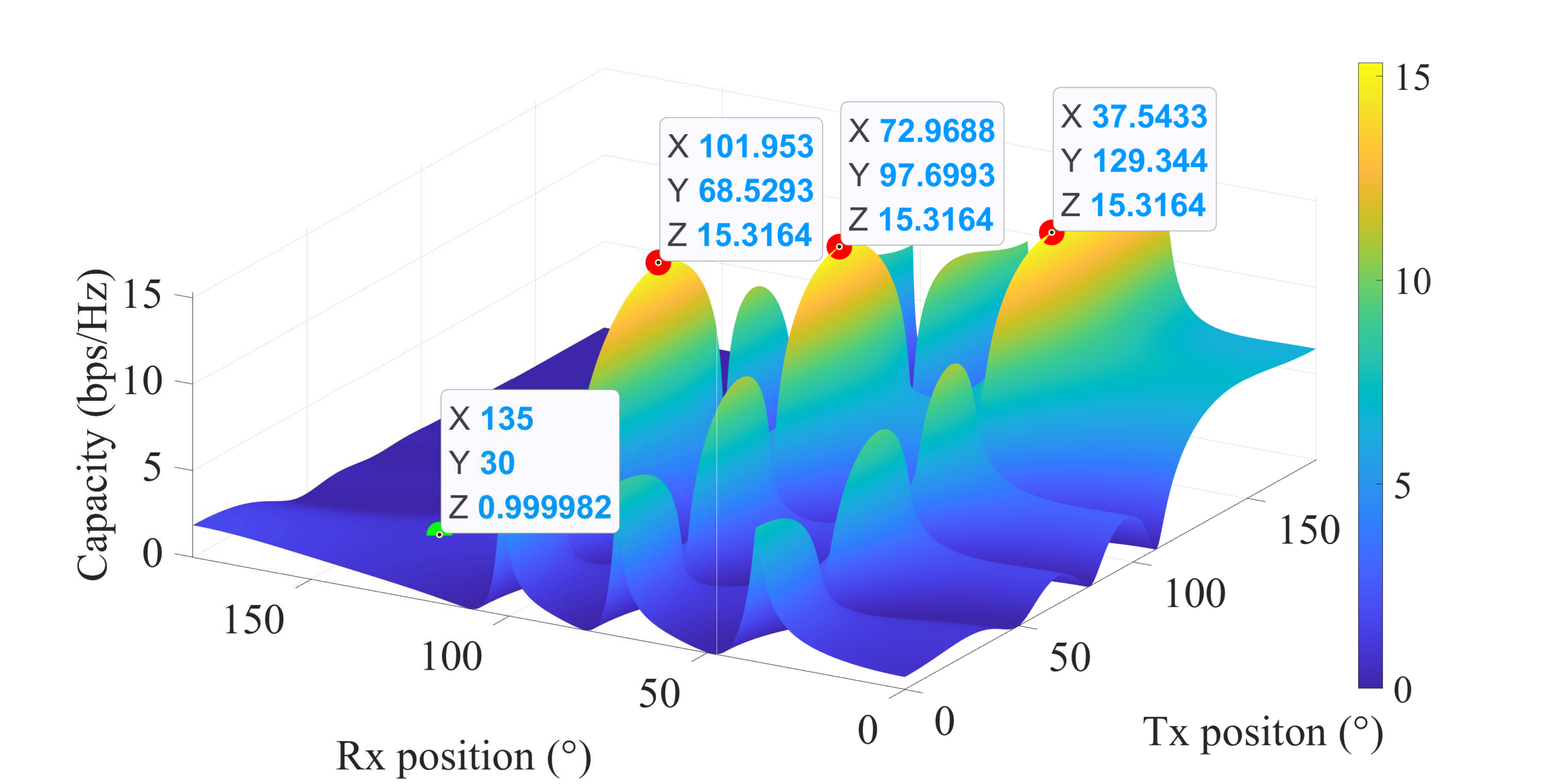}}\hspace{0em}%
  \subcaptionbox{$M=4$, $L=2$\label{fig3:b}}{\includegraphics[width=2.375in]{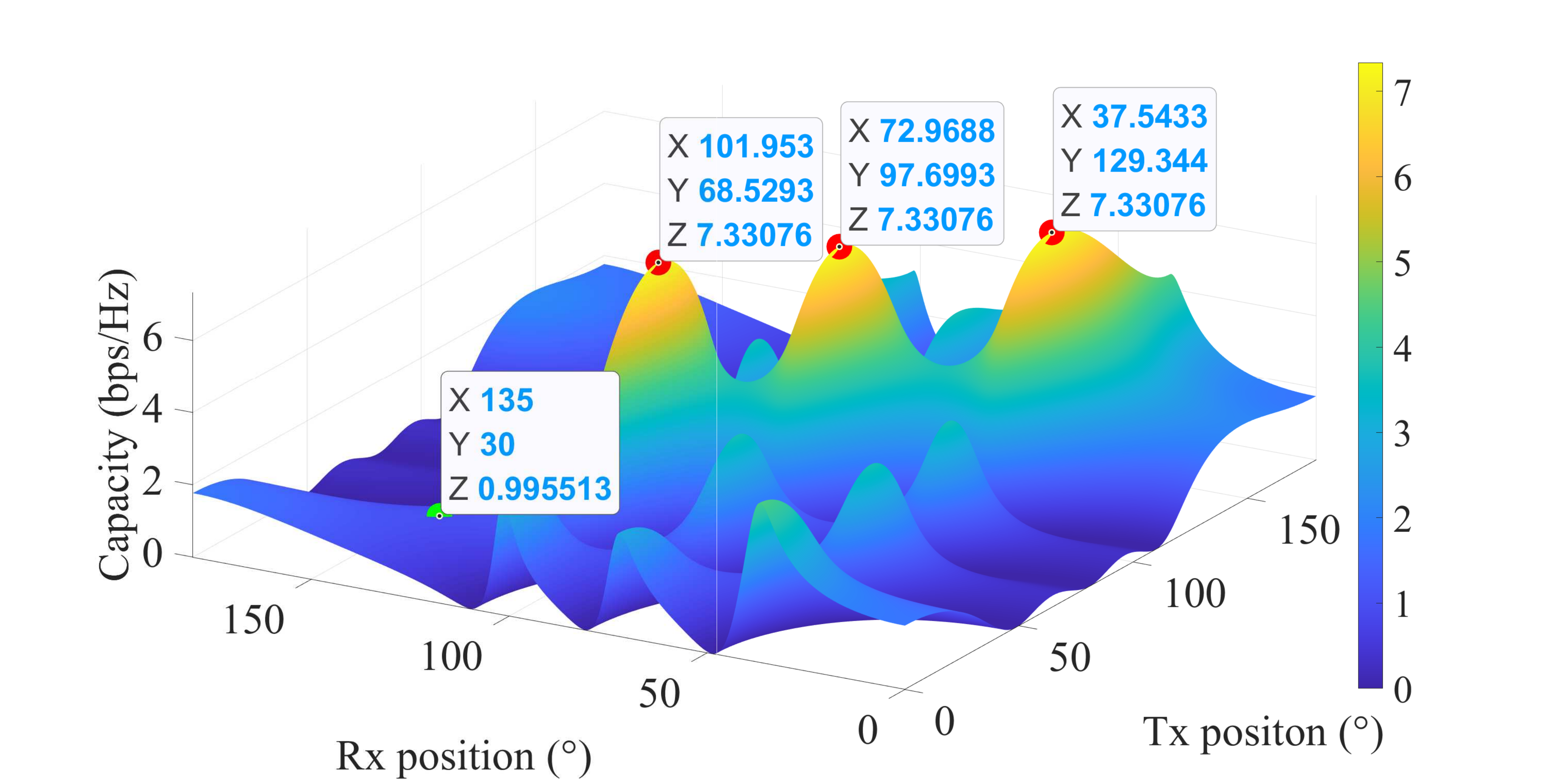}}\hspace{0em}%
 \subcaptionbox{$M=8$, $L=4$.\label{fig3:c}}{\includegraphics[width=2.375in]{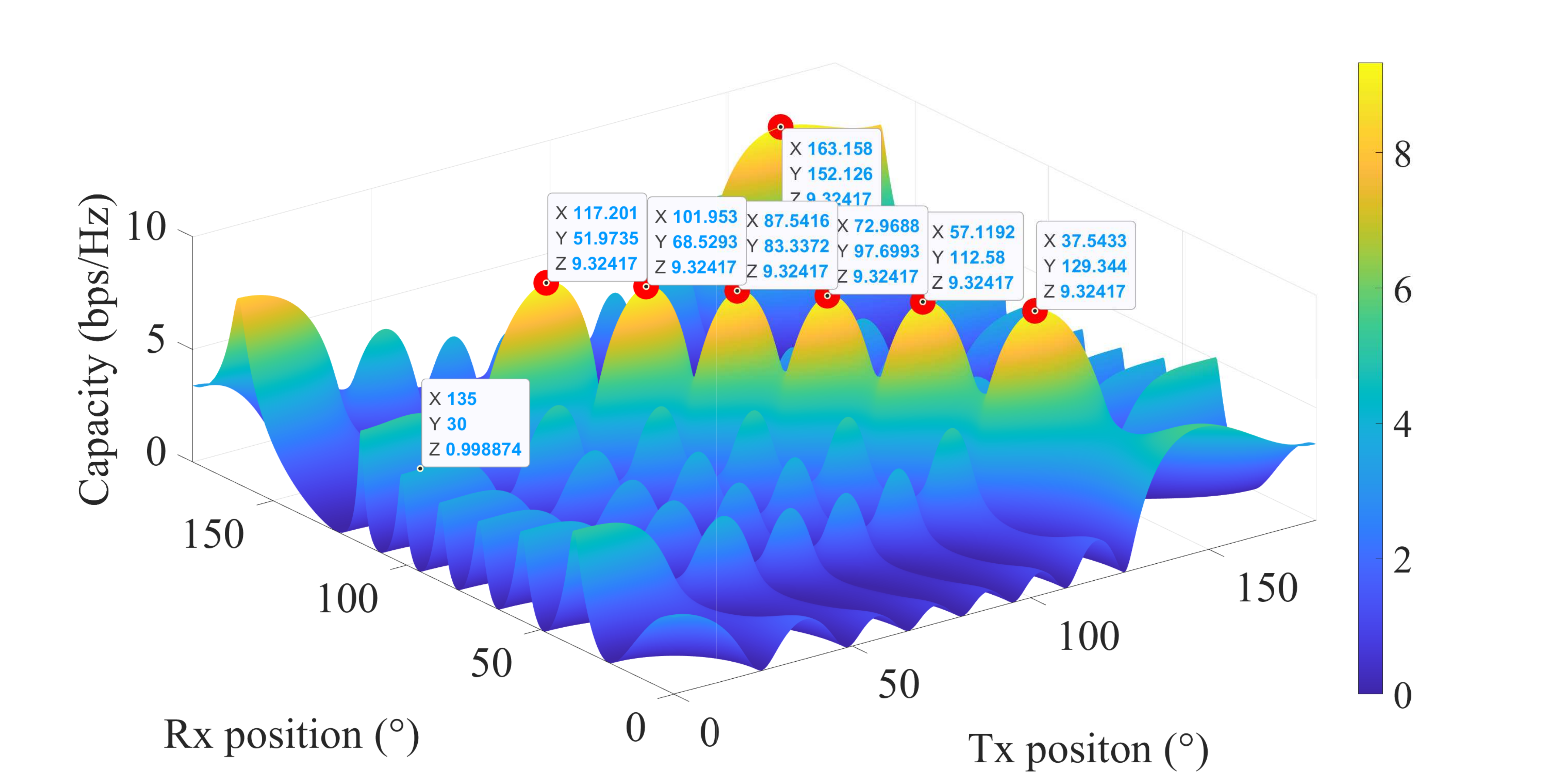}}
     \hfill
        \caption{The capacity vs pairs' position. The optimal positions for the 2nd pair are marked by red dots. Green dot marked the position of the 1st pair position.}
          \label{fig:3d_all}
\end{figure*}
To clarify the number of paths with different lengths in the network, we denote $N_\gamma$ as the number of transceiver pairs that their Eulerian paths have $\gamma$-length in the network. Thus, following Eq. (\ref{eq:c_lg}), we have
\begin{equation}\label{eq:cg_general}
   C_{MU,CG,(1,...,\Gamma)}=\sum_{\gamma=1}^{\Gamma}N_{\gamma}log(1+\frac{P_{T}M^{2\gamma}}{N_{0}}),
\end{equation}
\noindent which is the sum-rate upper bound of MUMOR transmission assisted by the CG network. Note that the number of total transceiver pair $N$ that can achieve interference-free transmission is a variable, where
\begin{equation}
  N=\sum_{\gamma=1}^{\Gamma}N_{\gamma}.
\end{equation}
Since there are multiple ways to decompose a CG into different number of Eulerian paths with different lengths, the value of $N_{\gamma}$, $\gamma=1,2,...,\Gamma$ are to be determined by a specific graph decomposition. To rewrite $N$ in a general expression, we decompose the CG into paths where all of their lengths are equal to $\gamma$. To ensure the upper bound is reached at maximal SE, all these Eulerian paths should pass through all edges. Note that, a class of graph decomposition problem is introduced here, which is determining if the CG network can be completely decomposed into paths of $\gamma$ length equally, which has been proven to be NP-complete \cite{dor1997graph}. Therefore, it is hard to determine $N_{\gamma}$ and write $N_{\gamma}$ in a general expression. 


In order to obtain a general expression of the sum-rate upper bound of CG network, we denote $\Lambda_{i}, i \in [1,N],$ as the path length for $i$-th pair transceiver, and we consider $\Lambda_i=\tau, i=1, ..., N$, and $\Gamma = \tau >1$, where $\tau$ is a specific value of path length. In addition, we denote \begin{equation}
 N_\tau=\frac{{K \choose 2 }}{\tau-1}=\frac{K(K-1)}{2(\tau-1)},
\end{equation}
 where ${K \choose 2 }$ is the total edges' number of a $K$-nodes CG. To completely decompose the CG, we should satisfy
 \begin{equation}
    N_\tau \in \mathbb{Z}\;,
 \end{equation}
as it is a necessary and sufficient condition for the existence of an edge-disjoint decomposition of a $K$-nodes CG into simple isomorphic paths consisting of $(\tau-1)$ edges each \cite{tarsi1983decomposition}. With $ N_\tau \in \mathbb{Z}$, the edge number of a CG can be equally divided up into paths with $\tau$-length. Thus, $ N_\tau$ is the multiplexing gain while the cascading power gain of a corresponding pair is $M^{2\tau}$. Following Eq. (\ref{eq:cg_general}), the sum-rate now becomes
\begin{equation}\label{eq:cg1}
   C_{MU,CG,(\tau)}=N_{\tau}log(1+\frac{P_{T}M^{2\tau}}{N_{0}}),
\end{equation}
\noindent which is the sum-rate upper bound for the MUMOR transmission for $N_\tau$ pairs transceivers with length of $\tau$. By combining the sum-rate upper bound of $1$-length paths, we have
\begin{multline}\label{eq:cg2}
   C_{MU,CG,(1,\tau)}=N_{\tau}log(1+\frac{P_{T}M^{2\tau}}{N_{0}}) +\\
   (KM-N_{\tau}\tau)log(1+\frac{P_{T}M^{2}}{N_{0}}),
\end{multline}
and now the upper bound is reached for $N=KM+N_\tau(1-\tau)$ pairs of transceiver. The form of second term can be derived similarly as to derive Eq. (\ref{eq:chain}).

For the sum-rate upper bound in a general case, i.e., path lengths are different for different pairs, the sum-rate upper bound can still be computed as long as the graph decomposition is determined. Then, the value of $N_\gamma$ is fixed, and the sum-rate upper bound can be computed using Eq. (\ref{eq:cg_general}). 

\subsection{The IRS Network In Null Graph}
When $\tau=1$, since the graph of the network has no edges, we can call it a null graph (NG). In this case, each IRS serves a local network in different cells and no edges connect any two IRS nodes. The sum-rate upper bound can be straightforwardly obtained from Theorem \ref{theo:3} as $C_{MU,NG}=KC_{Max}$, which is directly scaled by $K$-folds. Since each IRS node is isolated locally, inter-user interference is not induced. 


Also, as one proof has been shown in \cite{tse2005fundamentals} that leveraging Jensen's inequality, we know at low SNR, the sum-rate reaches an upper bound if equal decomposition is realized for the $K$ nodes CG with largest $\tau$. In addition, at high SNR, the upper bound is reached with $\tau=1$.

\section{Simulation}\label{section:RB}
\subsection{The Single IRS Optimal Capability}
In Fig.~\ref{fig:3d_all}, we consider the MRC solution of beamforming to illustrate the optimal transceiver position of IRS depicted in the Section IV, where MRC is optimal for the \nth{1} fixed pair, which is located at $\mathbf{a}_{C}(\phi_{in,1},\phi_{out,1})=(30\degree,135\degree)$ considering ULA shape's IRS. SNR is assumed to be $10$ dB.
According to the theorems proposed in this paper, we can analytically calculate the optimal available positions for the \nth{2} pair, where it can harvest maximal power gain from single IRS with nullified interference from the \nth{1} fixed pair. Analytically, these positions are $(68.53\degree,101.95\degree)$, $(97.70\degree,72.97\degree)$, and $(129.34\degree,37.54\degree)$, respectively when the relative length $L=2$. Fig.~\ref{fig:3d_all}(a), (b) and (c) show that theorems in Section IV accurately depict the optimal positions for other pairs. We can also observe that increasing the elements under the fixed-length $L$ will not change the optimal positions. All the optimal positions remain in the same place but only with higher power gain. Fig.~\ref{fig:3d_all}(c) shows that doubling $L$ also doubles the number of optimal positions, and $8$ pairs can be optimally supported in this case.


  \begin{figure}
\hspace*{0mm}
    \centering   
    \adjincludegraphics[width=90mm,height=50mm,right]{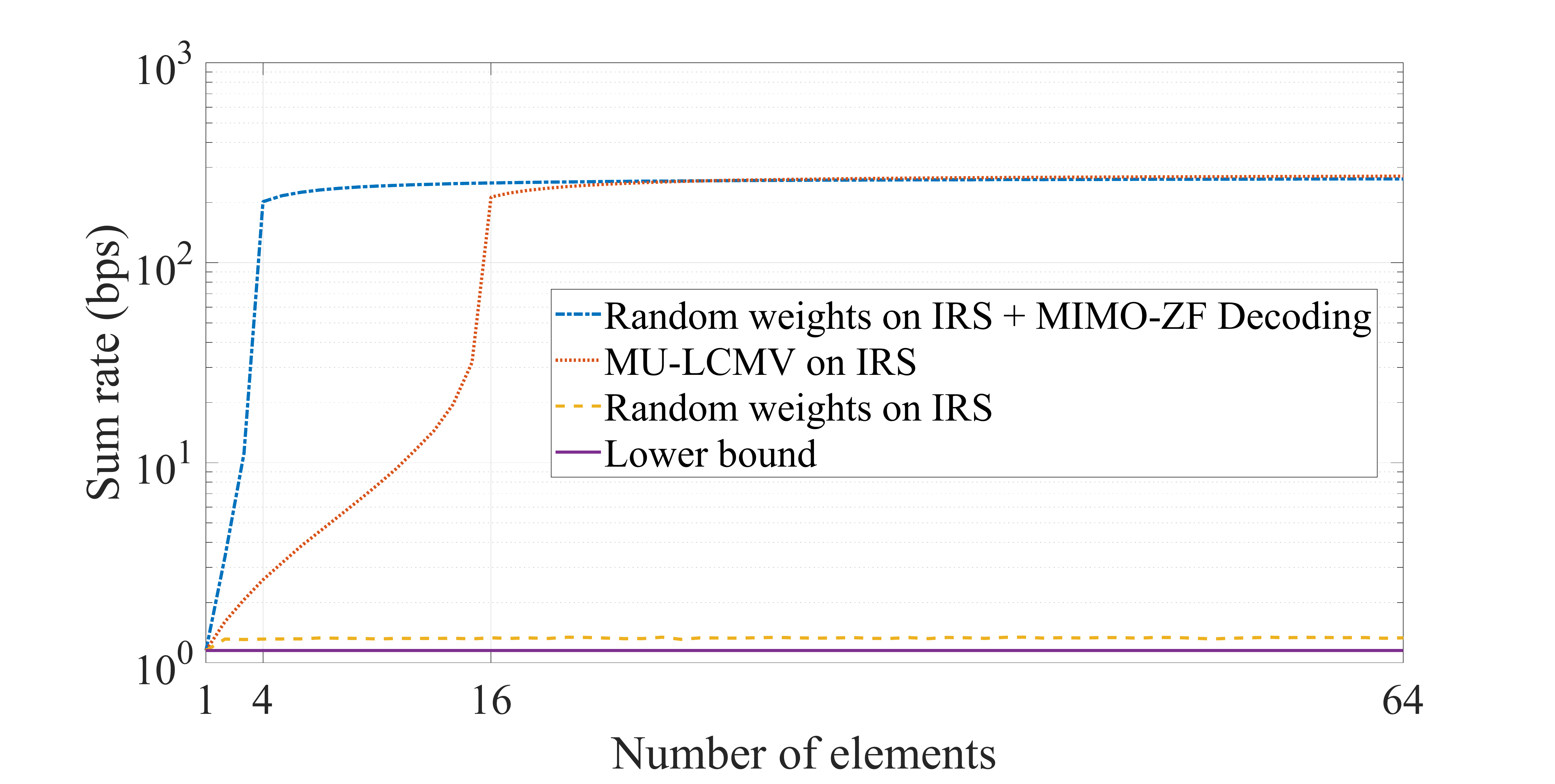}
    \caption{Sum-rates of three different transmission schemes changing with IRS elements number, given $N=4$.}
    \label{fig:equi_m}
\end{figure}
\subsection{The Single IRS Interference Suppressing}
To validate the interference-free transmission scheme is effective with only single IRS, we simulate $10000$ times the realization of three transmission schemes. The first one is the C-LoS channel with random weights on the single IRS. The second scheme still transmit through the C-LoS channel with random weights but a $4$ by $4$ joint decoding matrix at \texttt{Rxs}' side using the zero-forcing (ZF) algorithm is leveraged as a benchmark (though it may not be practically implemented). The third one transmit through the C-LoS channel with weights obtained by multi-user linearly constrained minimum variance (MU-LCMV) algorithm \cite{liu2022multi}, which can simultaneously support multiple streams by a single IRS. For a specific realization, $4$ pairs of transceivers are distributed uniformly around an IRS and transmit normalized power. Since ZF at \texttt{Rxs} causes the noise amplification of \texttt{Rxs} but MU-LCMV from IRS does not, the noise is neglected at \texttt{Rxs} for a fair comparison.

The sum-rates of these three schemes changes with the number of elements of a single IRS is shown in Fig.~\ref{fig:equi_m}. The lower bound of the sum-rate for $4$ pairs of transceivers is plotted for reference. It can be observed that with a relatively small amount of reflector elements, e.g., $4<M<16$, the IRS with MU-LCMV algorithm is less likely to outperform the traditional MIMO ZF-decoding scheme. At this point, the capability of IRS is less likely to manage the interference with a limited amount of elements. However, the IRS can suppress the interference effectively at $M=16$, where the sum-rate exhibits a jump. This is critical since the relation of $M=N^2$ in Lemma \ref{lemma:inff} is exactly satisfied. After that, the sum rate of MU-LCMV on the IRS also reaches a plateau and can have a equivalent performance with the ZF decoding scheme. Nevertheless, since the size of decoding matrix is fixed, with sufficiently large $M$ on the IRS, the MU-LCMV scheme can finally outperform the benchmark in terms of the power gain from the controlled channel.

\begin{figure}
\hspace*{0mm}
    \centering
    \adjincludegraphics[width=90mm,height=50mm,center]{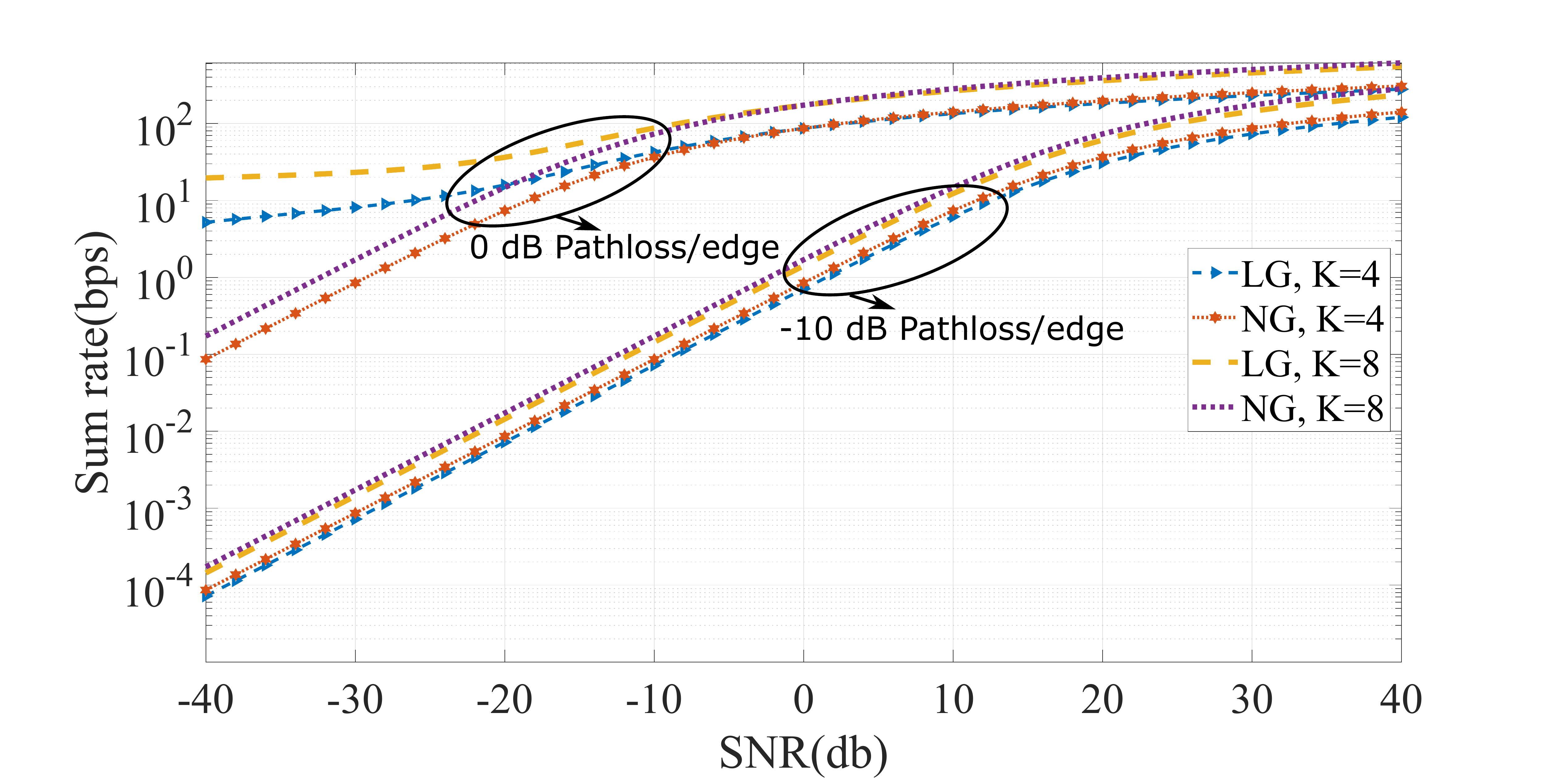}
    \caption{The sum-rate upper bound of LG and NG network with optimal condition.}
    \label{fig:chaindis}
\end{figure}

\subsection{The IRS Network Capability}
For illustrating the sum rate upper bound of networks, the pathloss is assumed to be $0$ dB while the scenario with pathloss of $-10$ dB per edge is also involved. As shown in Fig.~\ref{fig:chaindis}, the sum-rate upper bound of LG ($\Gamma=K$) and NG ($\Gamma=1$) are compared under different SNR, given $M=6$ and $K=4,8$. When pathloss is neglected, in a low SNR region, the sum-rate upper bound of the LG outperforms that of the NG due to the power gain from each C-LoS can be positively cascaded. In contrast, in a high SNR region, the sum-rate upper bound of the NG performs better than that of the LG due to larger spatial multiplexing gain is leveraged. However, given an apparent pathloss, the NG network achieves a better sum rate since each cascading of IRS node only cause larger loss on the cascaded power gain. Thus, the transmission leveraging the most FOR paths in the networks is preferred in this case. 

Fig.~\ref{fig:cap_cg} displayed the sum-rate upper bound of CG networks with different path lengths and IRS nodes, where $M=6$. For $K=4$, we can have $\Gamma=2,4$ while for $K=6$, we have $\Gamma=2,4,6$ such that the graph decomposition into Eulerian paths with equal length is complete. Note that, the sum rate upper bound of CG is dominated by the spatial multiplexing gain. Since the CG network can shape more FOR paths with less number of transceiver nodes leveraging edges of the CG, decomposing the CG with largest path length should result in the least number of MOR paths and hence the sum rate upper bound is also relating to the value of maximum order of reflections $\Gamma$. In addition, both Fig.~\ref{fig:chaindis} and Fig.~\ref{fig:cap_cg} verify the sum-rate of networks increases substantially with $K$ folds scaling, as we analyzed in Section V.



\begin{figure}
\hspace*{0mm}
    \centering
    \adjincludegraphics[width=90mm,height=50mm,center]{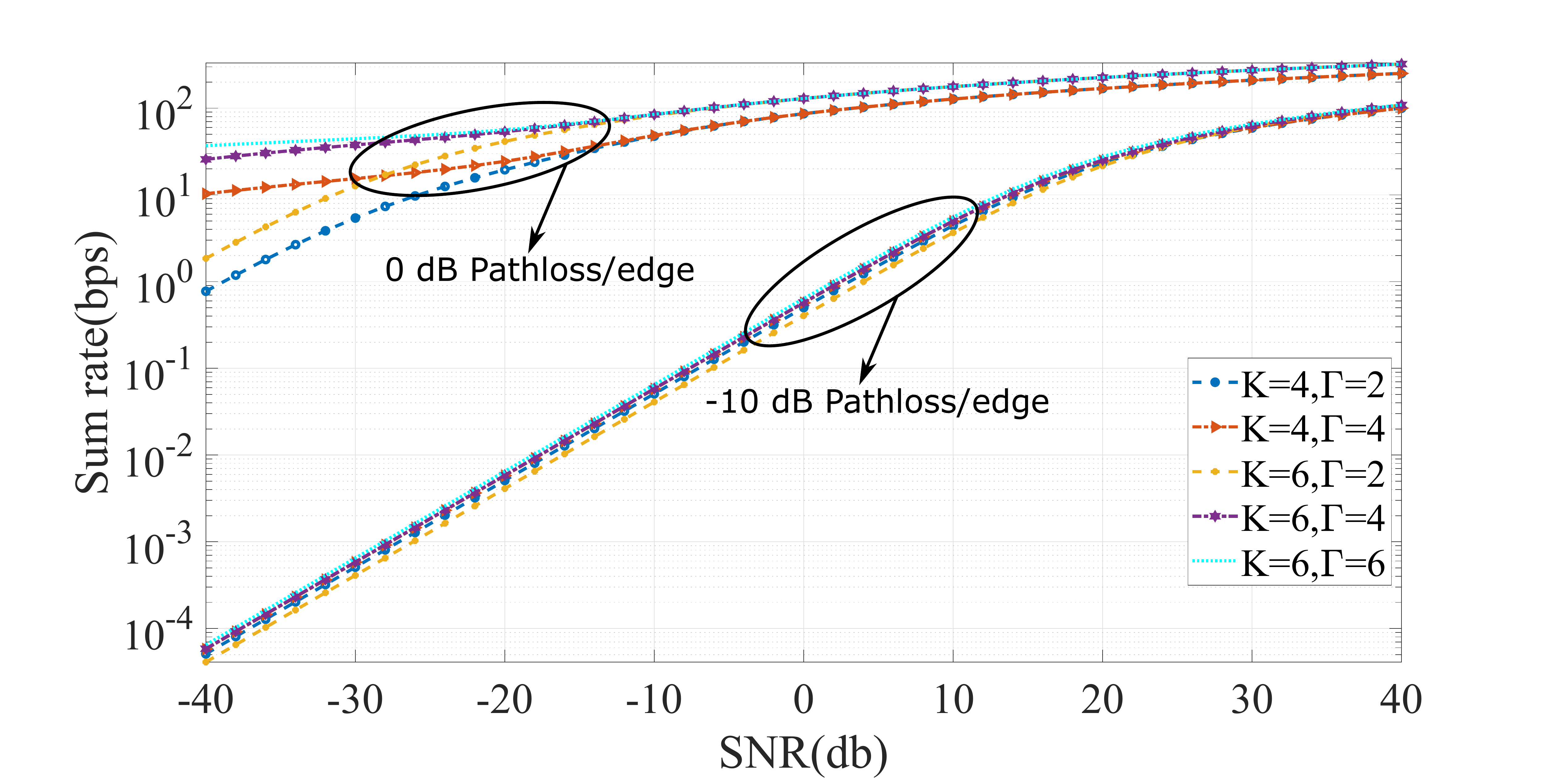}
    \caption{The sum-rate upper bound of MUMOR CG network with optimal condition.}
    \label{fig:cap_cg}
\end{figure} 

\section{Conclusion}
In this paper, we study the MUMOR transmission assisted by the IRS network. Firstly, we analytically establish a complete model of IRS network by permutationally combining two fundamental models. Secondly, the optimal condition to reach the sum-rate upper bound is derived, where the function of optimal positions for the transceivers is written in a closed form. In addition, we found that to sufficiently realize interference-free transmission, $M\geq N^2$ should be satisfied. Lastly, the sum-rate upper bound which can be provided by the IRS network is analyzed, where we specific topology can enhance the sum-rate with respect to different number of users and SNR. The simulation results verify our proposed theorems and indicate a promising $K$ folds scaling from the IRS network. 
\appendices 
\section{}\label{proof:ULA and URA optimal angle}
Lemma \ref{theo:opp} can be proved by analysing the IRS channel as a whole. I.e, we start by analysing the channel of a single transceiver pair assisted by a single IRS. By referring Eq. (\ref{eq:General channel}), we denote $\mathbf{H}_{I,1}=\mathbf{H}=\mathbf{A}_{out}^T\mathbf{W}\mathbf{A}_{in}$ for simplicity. Thus, the channel between \texttt{Tx}$_i$ and \texttt{Rx}$_i$ in $\mathbf{H}$ can be written as 
\begin{equation}\label{eq:siso channel}
\begin{split}
h_{ii}&=\mathbf{w}^{H}\mathbf{a}_{C}(\phi_{out,i},\phi_{in,i}).\\
   \end{split}
\end{equation}
We can observed that the diagonal terms in the matrix of Eq. (\ref{eq:General channel}) are signal gains for each \texttt{Rx} and these terms are required to be maximized. Other off-diagonal terms are the interference gain which should be minimized. Therefore, by calculating a optimal weights vector $\mathbf{w}$ such that the diagonal terms are maximized while nullifying off-diagonal terms, the optimal IRS based channel can be obtained and the optimal sum-rate can be achieved. Note that, if the single IRS is considered as the ULA or URA specification which has the characteristic of equal spacing between each elements, the optimal weights can be analytically obtained simply by MRC algorithm. Specifically, for ULA scenario and we let the i-th pair user locate at $\phi_{in,i}=\alpha_{i}\degree\,,\phi_{out,i}=\beta_{i}\degree$, Eq. (\ref{eq:siso channel}) can be rewritten as 
\begin{equation}
\label{eq:ulaop}
\begin{split}
h_{ii}&=\mathbf{w}^H\mathbf{a}_{C}(\phi_{out,i}\,,\phi_{in,i})=\sum_{m=0}^{M-1}w_{m}e^{-jkd(\cos\alpha_{i}+\cos\beta_{i})m },\\
\end{split}
\end{equation}
 where $k=\frac{2\pi}{\lambda}$ is the wave number, $d$ is the distance between each element and $\lambda$ is the carrier wavelength. The path-loss here is assumed to be a constant value. Thus, with unit power constraint on each IRS element, the weight on an IRS can then be expressed as
 \begin{equation}
w_{m}=e^{j\theta_{m}}\,,\theta_{m}\in(0,2\pi]\,,m=1,\dots,M\;\;.    
\end{equation}
As we can find, a necessary condition for $|h_{ii}|=M$ is that the weights need to guarantee each term in the summation in phase by writing the channel gain as 
\begin{equation}
\begin{split}
h_{ii}&=\mathbf{w}^H\mathbf{a}_{C}(\phi_{out,i}\,,\phi_{in,i})=\sum_{m=0}^{M-1}e^{-jkd(\cos\alpha_{i}+\cos\beta_{i}+\zeta_m)m },\\
\end{split}
\end{equation}
 where $\zeta_m$ is an arbitrary term comes from $\angle[w_m]$, the phase design on each element of IRS, we can observe the maximal value of $|h_{ii}|=M$ is guaranteed as long as 
 \begin{equation}\label{eq:optimaltx}
    kd(cos\alpha_{i}+cos\beta_{i}+\zeta_m)=2\pi n_1\,,n_1\in Z.
 \end{equation}
Denote $\Delta r=\frac{d}{\lambda}$, which is the normalized spacing between each element. We can compute the weight value on $m$-th element such as 
\begin{equation}
    \zeta_m=-\cos\alpha_{i}-\cos\beta_{i}+\frac{K}{\Delta r},
\end{equation}
to equalize the phase shifts. This is essentially the same to use MRC algorithm to calculate weights vector. Actual phase of weights can be obtained by $\theta_m=-\zeta_{m}kdm$. Then, after applying the result of MRC, since the weights have been determined, we can analyze other terms in the i-th column of matrix in equation (\ref{eq:General channel}) and write them as
\begin{equation}
\begin{split}
h_{ji}&=\mathbf{w}^H\mathbf{a}_{C}(\phi_{out,j}\,,\phi_{in,i})\\
&=\sum_{m=0}^{M-1}e^{-j2\pi\Delta r(\cos\beta_{j}-\cos\beta_{i}+\frac{K}{\Delta r})m }\;,\\
\end{split}
\end{equation}
where $\Delta r=\frac{d}{\lambda}$. Denote $f_{cc}=(\cos\beta_{j}-\cos\beta_{i}+\frac{K}{\Delta r})$ and $L=M\Delta r$ which are the variable in angular domain and normalized length of IRS. Therefore, $h_{ji}$ can be generalized as the beampattern and thus becomes a function of $f_{cc}$ 
\begin{equation}
\begin{split}
h_{ji}(f_{cc})&=\mathbf{w}^H\mathbf{a}_{C}(\phi_{out,j}\,,\phi_{in,i})\\
&=\sum_{m=0}^{M-1}e^{-j2\pi\Delta r(\cos\beta_{j}-\cos\beta_{i}+\frac{K}{\Delta r})m }\\
&=e^{-j\Delta{r}f_{cc}(M-1)}\frac{\sin(\pi f_{cc} L)}{\sin(\pi f_{cc} \frac{L}{M})}\;.
\end{split}
\end{equation}
 We can simply verify that $h_{ji}$ is a periodic function of $f_{cc}$ and the period is  $\frac{1}{\Delta r}$. If the period of $h_{ji}(f_{cc})$ is within the visible angular range which is $f_{cc}\in[-2,2]$ in this case, there can be $M-1$ other pairs of transceivers communicating at the same time. These pairs can use the same frequency of carrier since they are orthogonal in angular domain, which is shown in Fig.~\ref{fig:OSDM}. The nullifying point of $h_{ji}$ is also in the period of $\frac{1}{\Delta r}$, separated by $\frac{1}{L}$. Therefore, we can determine other \texttt{Rx}'s position $\beta_j$ such that there is no interference from the i-th \texttt{Tx} where the position can be calculated by
 \begin{equation}
\beta_{j}=\cos ^{-1}\left(\frac{j}{L}-\zeta_m-\cos \alpha_{i} \pm \frac{1}{\Delta r}\right)\;.
\end{equation}
These also means if other \texttt{Rxs} are standing in the same position as the nullifying position of \texttt{Tx}$_{i}$, there will be no interference from \texttt{Tx}$_{i}$, so other terms in the i-th column of channel matrix can be nullified.  In addition, since the weights have been calculated as $\zeta_m$ is set by first pair, given the position of \texttt{Rx}$_{j}$, we can calculate the optimal position of \texttt{Tx}$_{j}$ correspondingly leveraging the Eq. (\ref{eq:optimaltx}) which is 
\begin{equation}
\alpha_{j}=\cos ^{-1}\left(\frac{j}{L}-\zeta_m-\cos \beta_{i} \pm \frac{1}{\Delta r}\right)\;.
\end{equation}
\section{}\label{proof:ULA}
 To make the proof easy to follow, we assume $M=4$ and $N=2$, where $M$ is the number of elements on IRS and $N$ is the number of transceiver pairs. However, it is worth noting that this conclusion can be extended to arbitrary numbers of $N$ and $M$.  Following the definition in the manuscript, we have 
         \begin{equation} 
\mathbf{A}_{in}=\mathbf{A}=[\mathbf{a}(\phi_{in,1}),\mathbf{a}(\phi_{in,2})]=\begin{bmatrix}
  a_{11}  & a_{21}  \\
   a_{12} & a_{22}  \\
   a_{13} & a_{23} \\
   a_{14} & a_{24}  \\
\end{bmatrix},
\end{equation}
which is the steering matrix of incident direction toward IRS and $\mathbf{a}(\phi_{in,i}), i=1,2$ is the steering vector of incident direction on the IRS. Note that, this is also the channel from the \texttt{Txs} to the IRS. Similarly, we define the steering matrix of exit directions, which also is the channel from the IRS to the \texttt{Rxs}, as
        \begin{equation} 
\mathbf{A}_{out}=[\mathbf{a}(\phi_{out,1}),\mathbf{a}(\phi_{out,2})]=\mathbf{B}=\begin{bmatrix}
   b_{11} & b_{12} & b_{13} & b_{14}\\
   b_{21} & b_{22} & b_{23} & b_{24}\\
\end{bmatrix}^T ,
\end{equation}
where we change the notations of $\mathbf{A}_{out}$ by $\mathbf{B}$ for easy understanding.
The weight matrix of IRS is defined as $\mathbf{W}$, which is 
\begin{equation} 
\mathbf{W}=\begin{bmatrix}
   w_1^* & 0 &0 &0  \\
   0 & w_2^* &0 &0  \\
   0 &0 & w_3^* &0 \\
   0 &0 &0 & w_4^*  \\
\end{bmatrix}.
\end{equation}
By ignoring the noise term, we can write the received signal vector as 
\begin{multline}\label{eq:m2m}
 \hat{\mathbf{y}}_{r}=\mathbf{B}^T\mathbf{W}\mathbf{A}\mathbf{s}=\begin{bmatrix}
   \hat{y}_1 \\
   \hat{y}_2  \\
\end{bmatrix}=
\\ \begin{bmatrix}
   b_{11} & b_{12} & b_{13} & b_{14}\\
   b_{21} & b_{22} & b_{23} & b_{24}\\
\end{bmatrix}\!\!
\begin{bmatrix}
   w_1^* & 0 &0 &0  \\
   0 & w_2^* &0 &0  \\
   0 &0 & w_3^* &0 \\
   0 &0 &0 & w_4^*  \\
\end{bmatrix}\!\!
\begin{bmatrix}
  a_{11}  & a_{21}  \\
   a_{12} & a_{22}  \\
   a_{13} & a_{23} \\
   a_{14} & a_{24}  \\
\end{bmatrix}\!\!
\begin{bmatrix}
   s_1 \\
   s_2  \\
\end{bmatrix}
\;.    
\end{multline}
where the vector $\mathbf{s}=[s_1 \,\,s_2]^T$ is the vector of transmitted signal from \texttt{Tx}$_{1}$ and \texttt{Tx}$_2$. 
  \begin{figure}
  \centering
  \includegraphics[width=80mm,height=40mm]{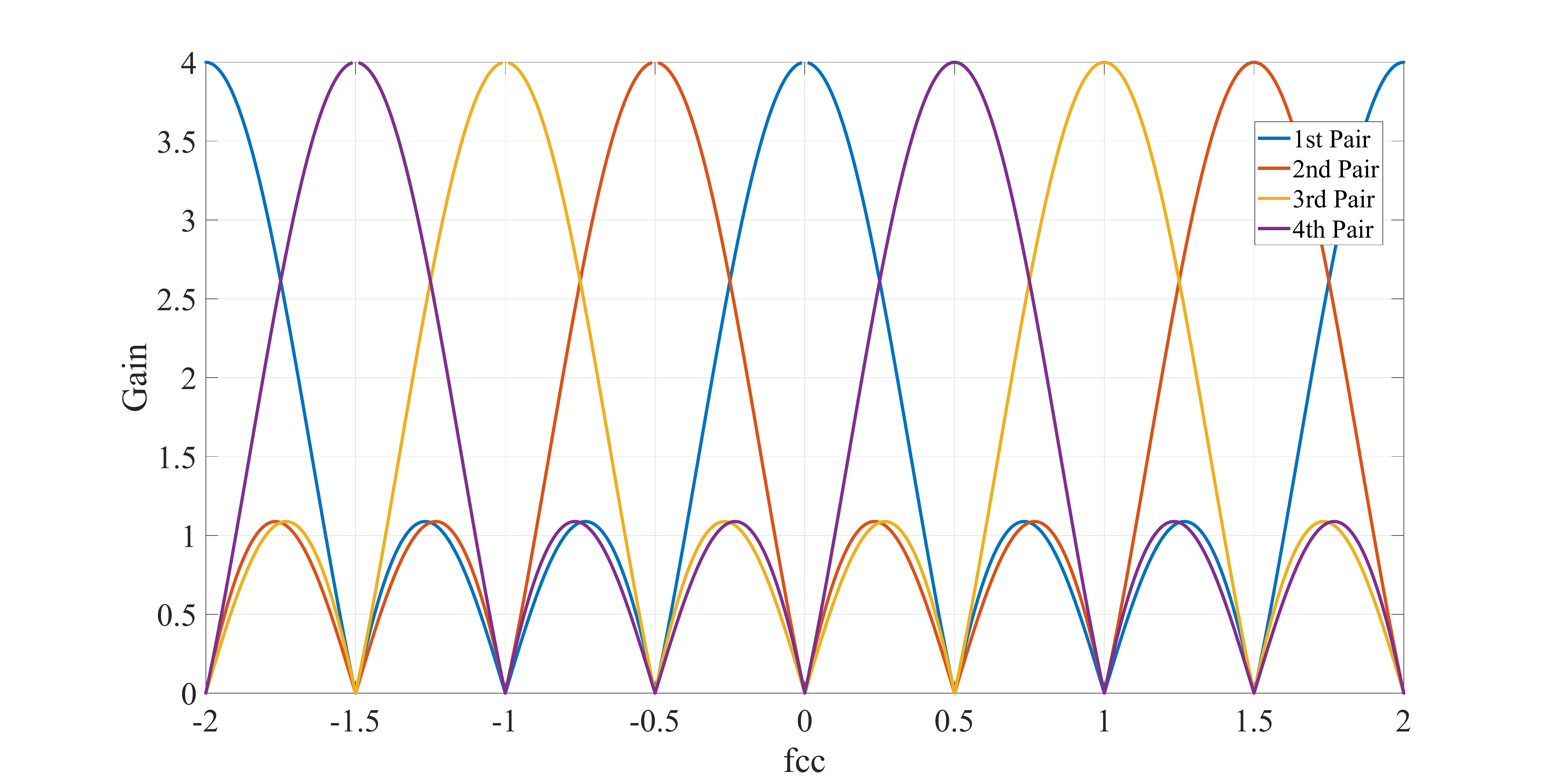}
  \caption{Optimal Spatial Multiplexing of $h_{1j}$,$h_{2j}$,$h_{3j}$ and $h_{4j}$,  M=4,d=$\frac{\lambda}{2}$, L=$2$.}
  \label{fig:OSDM}
\end{figure} 
 Next, some terms can be rewritten into a more regular form in order to have a channel expression which is similar to a traditional MIMO model. Therefore, by factorizing the term above in equation (\ref{eq:m2m}), we can have
\begin{equation}
    \begin{bmatrix}\label{big_mat}
   \hat{y}_1 \\
   \hat{y}_2  \\ 
\end{bmatrix}=\begin{bmatrix}
      [w_1^* & w_2^*& w_3^*& w_4^*]\begin{bmatrix}
   \mathbf{A}_{C,1}\\
\end{bmatrix}\begin{bmatrix}
  s_1\\
  s_2 \\
\end{bmatrix}\\
   [w_1^* & w_2^*& w_3^*& w_4^*]\begin{bmatrix}
\mathbf{A}_{C,2}
\end{bmatrix}\begin{bmatrix}
  s_1\\
  s_2 \\
\end{bmatrix}\\
\end{bmatrix},
\end{equation}
where
\begin{multline}
    \mathbf{A}_{C,1}=\begin{bmatrix}
\mathbf{a}_{C}(\phi_{out,1}\,,\phi_{in,1})    & \mathbf{a}_{C}(\phi_{out,1}\,,\phi_{in,2})
\end{bmatrix}=\\
\begin{bmatrix}
   b_{11}a_{11} & b_{11}a_{21}\\
   b_{12}a_{12} & b_{12}a_{22} \\
   b_{13}a_{13} & b_{13}a_{23}\\
   b_{14}a_{14} & b_{14}a_{24} \\
\end{bmatrix},
\end{multline}
and
\begin{multline}
\mathbf{A}_{C,2}=\begin{bmatrix}
\mathbf{a}_{C}(\phi_{out,2}\,,\phi_{in,1})    & \mathbf{a}_{C}(\phi_{out,2}\,,\phi_{in,2})
\end{bmatrix}=\\
\begin{bmatrix}
   b_{21}a_{11} & b_{21}a_{21}\\
   b_{22}a_{12} & b_{22}a_{22} \\
   b_{23}a_{13} & b_{23}a_{23}\\
   b_{24}a_{14} & b_{24}a_{24} \\
\end{bmatrix},
\end{multline}
         
        where $\mathbf{a}_{C}(\phi_{out,1}\,,\phi_{in,1}) =\mathbf{a}(\phi_{in,1})\odot \mathbf{a}(\phi_{out,1})$, $\mathbf{a}_{C}(\phi_{out,1}\,,\phi_{in,1})=\mathbf{a}(\phi_{in,2})\odot \mathbf{a}(\phi_{out,1})$. Therefore we can get the i-th user received signal as in \cite{Liu2019}
        
        \begin{equation}\label{eq:endentry}
    \hat{y}_{r,i}= \mathbf{w}^H\mathbf{A}_{C,i}\mathbf{s} + n_{i}\,, i=1, 2, ..., N\;\;.
\end{equation}
where the $n_i$ is the additive noise term at the each \texttt{Rx} and $\mathbf{w}$ is obtained by taking all the diagonal terms in $\mathbf{W}$, which is 
\begin{equation}
     \mathbf{w}=[w_1 \,\,\, w_2\,\,\, w_3\,\,\, w_4]^T.
\end{equation}
Note that that in equation (9), for different \texttt{Rxs}, their received signal is obtained along different steering matrix $\mathbf{A}_{C,i}$ but processed by the same weight vector $\mathbf{w}$. Due to $\mathbf{A}_{C,1}$ and $\mathbf{A}_{C,2}$ shares the same incident matrix, we can combine them further and move the difference on the two different steering matrices to the weight vector. Thus, through deviding $\mathbf{A}_{C,2}$ by $\mathbf{A}_{C,1}$ element-wisely, we can have matrix $\mathbf{C}$ which can be regarded as a factor of Hadamard product such that 

\begin{equation}
    \mathbf{A}_{C,1}\odot \mathbf{C}=\mathbf{A}_{C,2}, \text{and}  \,\mathbf{C}=\begin{bmatrix}
   \frac{b_{21}}{b_{11}} & \frac{b_{21}}{b_{11}}\\
  \frac{b_{22}}{b_{12}} & \frac{b_{22}}{b_{12}}\\
  \frac{b_{23}}{b_{13}} & \frac{b_{23}}{b_{13}}\\
  \frac{b_{24}}{b_{14}} & \frac{b_{24}}{b_{14}}\\
\end{bmatrix}.
\end{equation}
Actually, the term $\frac{b_{21}}{b_{11}}=e^{-jkd(cos\phi_{out,2}-\cos\phi_{out,1})0}$ for ULA case is complex constant where $k$ and $d$ are wave number and distance between elements respectively. Then, we can have $\frac{b_{2m}}{b_{1m}}=e^{-jkd(cos\phi_{out,2}-\cos\phi_{out,1})(m-1)}, m=1,2,...,M$. Although the terms' equivalence in the same column like $\frac{b_{21}}{b_{11}}=\frac{b_{22}}{b_{12}}=...=\frac{b_{24}}{b_{14}}$ can be achieved with the increasing of the iterative power term $(m-1)$ which means the steering matrices are same, we can assume that $d$ is small enough so that the overall complex term can not repeat in the period of itself and we can have $\frac{b_{21}}{b_{11}}\neq \frac{b_{22}}{b_{12}}\neq...\neq \frac{b_{24}}{b_{14}}$ given the directions of angle $\phi_{out,1} \neq \phi_{out,2}$. Next, we note that the columns of matrix $\mathbf{C}$ are same, then we can rewrite equation (\ref{big_mat}) as

\begin{equation}
        \begin{bmatrix}
   \hat{y}_1 \\
   \hat{y}_2  \\ 
\end{bmatrix}=\begin{bmatrix}\label{eq:c_mat}
       \mathbf{w}^H       \mathbf{A}_{C,1} 
  \mathbf{s}
\\
  \mathbf{w}^H       \mathbf{A}_{C,1} \odot \mathbf{C} \mathbf{s}\\
\end{bmatrix}=\begin{bmatrix}
       \mathbf{w}^H       \mathbf{A}_{C,1} 
  \mathbf{s}
\\
   \mathbf{w_C}^H\mathbf{A}_{C,1}\mathbf{s}\\
\end{bmatrix},
\end{equation}
\begin{equation}
\mathbf{w_C}=[w_{c1}\,\,\, w_{c2}\,\,\,...\,\,\,w_{c4}]=[w_1\frac{b_{21}}{b_{11}} \,\,\,w_2\frac{b_{22}}{b_{12}}\,\,\,...\,\,\,w_2\frac{b_{24}}{b_{14}}]^T.
\end{equation}
$\mathbf{w_C}$ is the equivalent vector for the second \texttt{Rx} $\hat{y}_2$ and we can know that it has a mapping relationship to $\mathbf{w}$, which is the unique characteristic in the IRS's model. Therefore, by combining the common term in equation (\ref{eq:c_mat}), we have
        \begin{multline}\label{eq:cc_mat}
    \hat{\mathbf{y}}_{r}=\begin{bmatrix}
       \mathbf{w}^H        
\\
   \mathbf{w_C}^H\\
\end{bmatrix}\begin{bmatrix}
   \mathbf{A}_{C,1}\mathbf{s}
\end{bmatrix}=
\\ \begin{bmatrix}
w_{1}^*      & w_{2}^* & w_{3}^* & w_{4}^*
\\
  w_{c1}^*      & w_{c2}^*& w_{c3}^*& w_{c4}^*\\
\end{bmatrix}
\begin{bmatrix}
   b_{11}a_{11} & b_{11}a_{21}\\
   b_{12}a_{12} & b_{12}a_{22} \\
   b_{13}a_{13} & b_{13}a_{23}\\
   b_{14}a_{14} & b_{14}a_{24} \\
\end{bmatrix}
\begin{bmatrix}
   s_1\\s_2
\end{bmatrix},
        \end{multline}
and by multiplying weight matrix with steering matrix, we have
\vspace{-1mm}
    \begin{equation}
           \hat{\mathbf{y}}_{r} = \begin{bmatrix}\label{eq:finalmat}
     \mathbf{w}^H\mathbf{a}_{C}(\phi_{out,1}\,,\phi_{in,1}) & \mathbf{w}^H\mathbf{a}_{C}(\phi_{out,1}\,,\phi_{in,2})\\
     \mathbf{w_C}^H\mathbf{a}_{C}(\phi_{out,1}\,,\phi_{in,1})& \mathbf{w_C}^H\mathbf{a}_{C}(\phi_{out,1}\,,\phi_{in,2})\\
     
\end{bmatrix}\begin{bmatrix}
   s_1\\s_2
\end{bmatrix}.
        \end{equation}
        
    
     To suppress the interference, we need to diagnolize the matrix in equation (\ref{eq:finalmat}). Namely, the weights vector $\mathbf{w}$ should satisfy  
\begin{equation}\label{eq:conflict2}
\begin{cases}
\mathbf{w}^H\mathbf{a}_{C}(\phi_{out,1}\,,\phi_{in,1})= \delta_{1}\\
\mathbf{w}^H\mathbf{a}_{C}(\phi_{out,1}\,,\phi_{in,2})=
0\\
\mathbf{w_{C}}^H\mathbf{a}_{C}(\phi_{out,1}\,,\phi_{in,1})=0\\
\mathbf{w_{C}}^H\mathbf{a}_{C}(\phi_{out,1}\,,\phi_{in,2})=\delta_{2}\\
\end{cases}\;,
\end{equation}
where $\delta_{1}$ and $\delta_{2}$ are non-zero values. Since $\mathbf{w_C}$  can be replaced by $\mathbf{w}$, we can present equation (\ref{eq:conflict2}) by using matrix as

\begin{equation}
  \begin{bmatrix}
  b_{11}a_{11} & b_{12}a_{12} & b_{13}a_{13} & b_{14}a_{14} \\
  b_{11}a_{21} & b_{12}a_{22} & b_{13}a_{23} & b_{14}a_{24} \\
  b_{21}a_{11} & b_{22}a_{12} & b_{23}a_{13} & b_{24}a_{14} \\
  b_{21}a_{21} & b_{22}a_{22} & b_{23}a_{23} & b_{24}a_{24} \\
    \end{bmatrix}\begin{bmatrix}
    w_1^* \\
    w_2^* \\
    w_3^* \\
    w_4^* \\
    \end{bmatrix}=\begin{bmatrix}
    \delta_1 \\
    0 \\
    0 \\
    \delta_2 \\
    \end{bmatrix}.
\end{equation}
As the matrix on the left-hand side is full rank which is assured by the assumption above, $4$ linear equations with $4$ unknowns can be solved with a non-zero solution. Moreover, by increasing the element number such that $M>>N^2$, the solution space will be further enlarged. Thus, there must be multiple non-zero solutions to achieve the diagonalization of the matrix in the equation (\ref{eq:finalmat}). In this case, the weights $\mathbf{w}$ and $\mathbf{w_C}$ can be nearly orthogonal to each other. As a result, the equivalence between traditional MIMO and IRS is established, and the interference can be suppressed among multiple transceiver pairs.

\bibliographystyle{ieeetr}
\bibliography{citation}

\end{document}